\documentclass[11pt,twoside,a4paper,notitlepage]{article}

\usepackage[a4paper,margin=2.5cm]{geometry}

\makeatletter{}

\usepackage{hyperref}
\usepackage[utf8]{inputenc}

\usepackage{stmaryrd}

\usepackage{amsmath,amssymb,stmaryrd,url,nicefrac}
\usepackage{xparse}
\usepackage{enumerate}
\usepackage{mathdots}
\usepackage{graphicx}
\usepackage{mathtools}
\usepackage{pdfsync}
\usepackage{authblk}

\usepackage{amsthm}
\newtheoremstyle{mythm}{}{}{\itshape}{}{\bfseries}{.}{.5em}{\thmname{#1}~\thmnumber{#2}\ifthenelse{\equal{\thmnote{#3}}{}}{}{~(\thmnote{#3})}}

\newtheoremstyle{mydefn}{}{}{\upshape}{}{\bfseries}{.}{.5em}{\thmname{#1}~\thmnumber{#2}\ifthenelse{\equal{\thmnote{#3}}{}}{}{~(\thmnote{#3})}}

\newtheoremstyle{myremark}{}{}{\upshape}{}{\itshape}{.}{.5em}{\thmname{#1}~\thmnumber{#2}\ifthenelse{\equal{\thmnote{#3}}{}}{}{~(\thmnote{#3})}}

\theoremstyle{mythm}
\newtheorem{theorem}{Theorem}[section]
\newtheorem{lemma}[theorem]{Lemma}

\newtheorem{corollary}[theorem]{Corollary}

\theoremstyle{mydefn}
\newtheorem{definition}[theorem]{Definition}
\newtheorem{example}[theorem]{Example}
\theoremstyle{myremark}
\newtheorem{remark}[theorem]{Remark}
\theoremstyle{mythm}

\makeatletter
\newcommand*{\transpose}{{^{\mathpalette\@transpose{}}}}
\newcommand*{\@transpose}[2]{\raisebox{\depth}{$\m@th#1\intercal$}}
\makeatother

\newcommand{\bN}{{{\mathbb N}}}
\newcommand{\bZ}{{{\mathbb Z}}}
\newcommand{\cA}{{\ensuremath{\mathcal{A}}}}
\newcommand{\cF}{{\ensuremath{\mathcal{F}}}}
\newcommand{\cT}{{\ensuremath{\mathcal{T}}}}
\newcommand{\cB}{{\ensuremath{\mathcal{B}}}}
 \DeclareMathOperator{\dist}{dist}
\newcommand{\sem}[1]{\llbracket #1 \rrbracket}

 \newcommand{\quants}[1]{\ensuremath{\mathbb{#1}}}
 \newcommand{\Ps}{\quants{P}}
 \DeclareMathOperator{\sph}{sph}

 \newcommand{\ov}[1]{\ensuremath{\bar{#1}}}

 \newcommand*{\size}[1]{\ensuremath{|\!|#1|\!|}}

 \newcommand{\True}{\ensuremath{\textit{true}}}
 \newcommand{\False}{\ensuremath{\textit{false}}}

 \DeclareMathOperator{\ar}{ar}

 \newcommand*{\free}{\textup{free}}

 \DeclareMathOperator{\countr}{d_{\#}}

\newcommand{\nc}[1]{\newcommand{#1}}
\newcommand{\rnc}[1]{\renewcommand{#1}}

\nc{\true}{\True}
\nc{\false}{\False}

\nc{\deff}{:=}

\rnc{\leq}{\leqslant}
\rnc{\geq}{\geqslant}

\rnc{\le}{\leqslant}
\rnc{\ge}{\geqslant}

\rnc{\phi}{\varphi}

\nc{\NN}{\bN}
\nc{\NNpos}{\ensuremath{\NN_{\geq 1}}}
\nc{\ZZ}{\bZ}

\nc{\Structure}[1]{\ensuremath{\mathcal{#1}}}
\nc{\A}{\cA}
\nc{\B}{\cB}
\nc{\F}{\cF}
\nc{\T}{\cT}
\nc{\I}{\Structure{I}}

\nc{\isom}{\ensuremath{\cong}}
\nc{\isomorph}{\isom}
\nc{\equivd}{\ensuremath{\equiv_d}}

\nc{\bigoh}{\mathcal{O}}
\nc{\bigOh}{\bigoh}
\nc{\littleoh}{o}
\nc{\littleOh}{\littleoh}

\nc{\set}[1]{\ensuremath{\{ #1 \}}}
\nc{\setc}[2]{\set{#1 : #2}}
\nc{\bigset}[1]{\ensuremath{\big\{ #1 \big\}}}
\nc{\bigsetc}[2]{\bigset{#1 : #2}}
\nc{\setsize}[1]{\ensuremath{|#1|}}
\nc{\bigsetsize}[1]{\ensuremath{\big|#1\big|}}
\nc{\Setsize}[1]{\bigsetsize{#1}}

\nc{\sphere}[2]{\ensuremath{\sph_{#1}(#2)}}

\nc{\und}{\ensuremath{\wedge}}
\nc{\Und}{\ensuremath{\bigwedge}}
\nc{\oder}{\ensuremath{\vee}}
\nc{\Oder}{\ensuremath{\bigvee}}
\nc{\nicht}{\ensuremath{\neg}}
\nc{\impl}{\ensuremath{\to}}
\nc{\gdw}{\ensuremath{\leftrightarrow}}

\nc{\MSO}{\ensuremath{\textup{MSO}}}

\nc{\FO}{\ensuremath{\textup{FO}}}
\nc{\FOC}{\ensuremath{\textup{FOC}}}

\nc{\FOunC}{\ensuremath{\FOC_1}}
\nc{\FOCN}{\ensuremath{\textup{FOCN}}}

\nc{\FOCP}{\ensuremath{\FOC(\Ps)}}
\nc{\FOunCP}{\ensuremath{\FOunC(\Ps)}}
\nc{\FOCNP}{\ensuremath{\FOCN(\Ps)}}

\nc{\VARS}{\ensuremath{\textsf{\upshape vars}}}
\nc{\NVARS}{\ensuremath{\textsf{\upshape nvars}}}

\nc{\Count}[2]{\ensuremath{\# {#1}.{#2}}}
\nc{\quant}[1]{\ensuremath{\textsf{\upshape #1}}}
\rnc{\P}{\quant{P}}

\nc{\query}[2]{\ensuremath{\texttt{\{} \, #1\, : \, #2 \,\texttt{\}}}}

\nc{\emptyword}{\ensuremath{\varepsilon}}
\nc{\emptytuple}{\ensuremath{()}}

\rnc{\S}{\ensuremath{\mathcal{S}}} 

\nc{\neighb}[3]{\ensuremath{N_{#1}^{#2}(#3)}}
\nc{\nbset}[2]{\ensuremath{N_{#1}^{#2}}}
\nc{\nrA}[1]{\ensuremath{\neighb{r}{\A}{#1}}}

\nc{\nrAStrich}[1]{\ensuremath{\neighb{r}{\B}{#1}}}

\nc{\nrT}[1]{\ensuremath{\neighb{r}{\T}{#1}}}
\nc{\nRT}[1]{\ensuremath{\neighb{R}{\T}{#1}}}

\nc{\Neighb}[3]{\ensuremath{\mathcal{N}_{#1}^{#2}(#3)}} 
\nc{\NrA}[1]{\ensuremath{\Neighb{r}{\A}{#1}}} 

\nc{\NRA}[1]{\ensuremath{\Neighb{R}{\A}{#1}}} 
\nc{\NrB}[1]{\ensuremath{\Neighb{r}{\B}{#1}}} 

\nc{\NRB}[1]{\ensuremath{\Neighb{R}{\B}{#1}}} 

\nc{\NrDStrich}[1]{\ensuremath{\Neighb{r}{\B}{#1}}} 

\nc{\NrT}[1]{\ensuremath{\Neighb{r}{\T}{#1}}} 

\nc{\Types}[3]{\ensuremath{\mathcal{T}_{#1}^{\sigma,#2}(#3)}}

\nc{\Typesrdk}{\Types{r}{d}{k}}
\nc{\Typesrd}[1]{\Types{r}{d}{#1}}
\nc{\Typesd}[2]{\Types{#1}{d}{#2}}
\nc{\isotypes}[2]{\ensuremath{\mathfrak{T}_{#1}^{#2}}}

\nc{\Typeslist}[3]{\ensuremath{\mathcal{L}_{#1}^{\sigma,#2}(#3)}}
\nc{\Typeslistrdk}{\Typeslist{r}{d}{k}}
\nc{\Typeslistrd}[1]{\Typeslist{r}{d}{#1}}
\nc{\Typeslistd}[2]{\Typeslist{#1}{d}{#2}}
\nc{\TypeslistRdk}{\Typeslist{R}{d}{k}}
\nc{\TypeslistRhatdk}{\Typeslist{\hat{R}}{d}{k}}
\nc{\TypeslistRd}[1]{\Typeslist{R}{d}{#1}}
\nc{\TypeslistRhatd}[1]{\Typeslist{\hat{R}}{d}{#1}}
\nc{\TypeslistRStrichdk}{\Typeslist{R'}{d}{k}}
\nc{\TypeslistRStrichd}[1]{\Typeslist{R'}{d}{#1}}

\nc{\type}{\ensuremath{\tau}}

\nc{\inducedSubStr}[2]{\ensuremath{#1[#2]}}

\nc{\mytime}{\ensuremath{\textit{time}^{\sigma,d}_{n,k,r}}}

\nc{\Class}{\ensuremath{\mathcal{C}}}
\nc{\GraphclassH}{\ensuremath{\mathcal{H}}}
\nc{\GraphclassG}{\ensuremath{\mathcal{G}}}
\nc{\Graphclass}{\ensuremath{\mathcal{G}}}

\nc{\ComplexityClassFont}[1]{\ensuremath{\textsc{#1}}}
\nc{\ACzero}{\ensuremath{\ComplexityClasFont{AC}^0}}
\nc{\PTIME}{\ensuremath{\ComplexityClassFont{Ptime}}}
\nc{\NP}{\ensuremath{\ComplexityClassFont{NP}}}
\nc{\PSPACE}{\ensuremath{\ComplexityClassFont{Pspace}}}

\nc{\FPL}{\ensuremath{\ComplexityClassFont{fpl}}}
\nc{\FPT}{\ensuremath{\ComplexityClassFont{fpt}}}
\nc{\AWstern}{\ensuremath{\ComplexityClassFont{AW}[*]}}
\nc{\Wone}{\ensuremath{\ComplexityClassFont{W}[1]}}

\nc{\sharpP}{\ensuremath{\#\ComplexityClassFont{P}}}
\nc{\sharpWone}{\ensuremath{\#\ComplexityClassFont{W}[1]}}

\nc{\lrounds}{\ensuremath{\rho}}

\newcommand{\cen}{\operatorname{cen}}

\newcommand{\cG}{{\mathcal G}}
\newcommand{\cX}{{\mathcal X}}

\newcounter{claimcounter}

\usepackage[textwidth=2cm]{todonotes}
\newcommand{\NewText}[1]{{#1}}
\newenvironment{NewTextBlock}{}

\newcommand{\FOplus}{\ensuremath{\FO^+}}
\newcommand{\myrho}{\delta}

\newcommand{\myp}{\ensuremath{\ell}}

\newcommand{\rHalbe}{\ensuremath{\frac{r}{2}}}

\begin{document}

\title{First-Order Query Evaluation with Cardinality Conditions}
\author[1]{Martin Grohe}
\author[2]{Nicole Schweikardt\thanks{Funded by the Deutsche
                  Forschungsgemeinschaft (DFG, German Research Foundation) -- SCHW 837/5-1.}}
\affil[1]{RWTH Aachen University}
\affil[2]{Humboldt-Universit\"at zu Berlin}

\maketitle

\begin{abstract}
We study an extension of first-order logic that allows to express
cardinality conditions in a similar way as SQL's COUNT operator. The
corresponding logic $\FOCP$ was introduced by Kuske and Schweikardt
\cite{KS_LICS17}, who showed that query evaluation for this logic is
fixed-parameter tractable on classes of databases of bounded degree.

In the present paper, we first show that the fixed-parame\-ter
tractability of $\FOCP$ cannot even be generalised to very simple
classes of databases of unbounded degree such as unranked trees or
strings with a linear order relation. 

Then we identify a fragment $\FOunCP$ of $\FOCP$ which is still
sufficiently strong to express standard applications of SQL's COUNT
operator. 
Our main result shows that query evaluation for $\FOunCP$ is
fixed-parameter tractable on nowhere dense classes of databases. 
\end{abstract}

\section{Introduction}\label{sec:intro}

Query evaluation is one of the most fundamental tasks of a database
system.
A large amount of the literature in database theory and the related field
of finite or algorithmic model theory is devoted to 
designing efficient query evaluation algorithms and to
pinpointing the exact computational complexity of the task.
The query languages that have received the most attention 
are the conjunctive queries and the more expressive relational calculus.
The latter is usually viewed as the ``logical core'' of SQL, and is
equivalent to first-order logic $\FO$.
Here, one identifies a database schema and a relational database of
that schema with a relational signature
$\sigma$ and a finite
$\sigma$-structure $\A$.

Apart from computing the entire query result, 
the query evaluation tasks usually studied are
\emph{model-check\-ing} (check if the answer $q(\A)$ of a Boolean query
$q$ on a database $\A$ is ``yes'') and
\emph{counting} (compute the number $|q(\A)|$ of tuples that belong to
the result $q(\A)$ of a non-Boolean query $q$ on a database $\A$);
the counting problem is also relevant as the basis of computing
probabilities. Such a task is regarded to be tractable for a query language
$\textup{L}$ on a class $\Class$ of databases if it can be solved in
time $f(k){\cdot}n^{c}$ for an arbitrary function $f$ and a constant
$c$, where $k$ is the size of the input query $q\in\textup{L}$ and $n$ the size of
the input database $\A\in\Class$.
The task then is called \emph{fixed-parameter tractable}
(fpt, or ``in $\FPT$''), and \emph{fixed-parameter linear} (fpl, or ``in
$\FPL$'') in case that $c=1$.

It is known that on unrestricted data\-bases model-checking 
is $\Wone$-hard
for conjunctive queries \cite{DBLP:journals/jcss/PapadimitriouY99},
and the counting problem is
$\sharpWone$-hard
already for
acyclic conjunctive queries \cite{DBLP:conf/icdt/DurandM13}.
This means that under reasonable complexity theoretic assumptions,
both problems are unlikely to be in $\FPT$.

A long line of research has focused on identifying restricted classes of
databases on which query evaluation is fixed-parameter
tractable for conjunctive queries, $\FO$, or
extensions of $\FO$.
For example, model-checking and counting for $\FO$ (even, for monadic
second-order logic) is 
in $\FPL$ on classes of bounded tree-width
\cite{DBLP:journals/iandc/Courcelle90, DBLP:journals/jal/ArnborgLS91}.
Model-checking and counting for $\FO$ are 
in $\FPL$ on classes of bounded degree \cite{See96, fri04}, 
in $\FPL$ on planar graphs and in $\FPT$ on classes of bounded local tree-width
\cite{FriG01,fri04}, and
in $\FPL$ on classes of bounded expansion
\cite{DBLP:conf/focs/DvorakKT10,kazseg13}.

Grohe, Kreutzer, and Siebertz
\cite{grokresie14} recently provided an $\FPT$ model-checking
algorithm for $\FO$ on classes of databases that are effectively
  nowhere dense.
This gives a fairly complete characterisation of the tractability
frontier for $\FO$ model-checking, as it is known that under
reasonable complexity theoretic assumptions, any subgraph-closed class
that admits an $\FPT$-algorithm for $\FO$ model-checking has to be
nowhere dense \cite{Kreutzer11-AMT-FAMT,DBLP:conf/focs/DvorakKT10}.
The notion of nowhere dense classes was introduced by {Ne\v{s}et{\v r}il} and
Ossona de Mendez
\cite{DBLP:journals/jsyml/NesetrilM10} as a formalisation of classes
of ``sparse'' graphs. The precise definition of this notion will be
relevant in this paper only in Section~\ref{sec:nd}; for now it should
suffice to note that the notion is fairly general, subsumes
all classes of databases mentioned above, and there
exist nowhere dense classes that do not belong to any of those classes.

The counting problem on nowhere dense classes is known to be
in $\FPT$ for purely existential $\FO$ \cite{nesoss12}, but 
no extension to full $\FO$ is known \cite{DBLP:conf/icdt/SegoufinV17}.
Here, we obtain this extension as an immediate consequence of our
technical main result. 
We study an extension of $\FO$ that allows to express
cardinality conditions in a similar way as SQL's COUNT operator. The
corresponding logic $\FOCP$ was introduced by Kuske and Schweikardt
\cite{KS_LICS17}, who showed that 
model-checking and counting for this logic is
fixed-parameter linear on classes of databases of bounded degree.
The starting point for the work presented in this paper was the question
whether this result can be extended to other ``well-behaved'' classes
of databases, such as the classes mentioned above.

Our first result is that the fixed-parameter
tractability of $\FOCP$ cannot even be generalised to very simple
classes of databases of unbounded degree such as unranked trees or
strings with a linear order relation. 
Then, we identify a fragment $\FOunCP$ of $\FOCP$ which still extends
$\FO$ and is
sufficiently strong to express standard applications of SQL's COUNT
operator. 
Our main result shows that model-checking and counting for $\FOunCP$ is
in $\FPT$ on nowhere dense classes of databases. 
More precisely, for any effectively nowhere dense class $\Class$ of
databases
we present an algorithm that
solves the model-checking problem and the counting problem in time 
$f(k,\epsilon){\cdot}n^{1+\epsilon}$ for a computable function $f$ and
any $\epsilon>0$, where $k$ is the size of the input query
$q\in\FOunC(\Ps)$ and $n$ is the size of the input database $\A\in\Class$.
Algorithms with such performance bounds are often called
\emph{fixed-parameter almost linear}.
This generalises the result of \cite{grokresie14} from $\FO$ to
$\FOunCP$ and solves not only the model-checking but also the counting
problem.

Our proof proceeds as follows. First, we reduce the query
evaluation problem for $\FOunCP$ to 
the counting problem for rather restricted $\FO$-formulas
(Section~\ref{sec:normalform}).
Combining this with the results on $\FO$-counting mentioned above,
we immediately obtain an $\FPT$-algorithm for
$\FOunCP$ on planar graphs and 
classes of bounded local tree-width \cite{fri04},  
of bounded expansion \cite{kazseg13}, 
and of locally bounded expansion \cite{DBLP:conf/icdt/SegoufinV17}.
For nowhere dense classes, though, it is not so
easy to generalise the $\FO$ model-checking 
algorithm of \cite{grokresie14} to solve the counting problem.
For this, we
\NewText{use a recently obtained ``rank-preserving Gaifman normal form
theorem'' of \cite{GS2026}
and transfer it to counting terms (Section~\ref{sec:Locality}),
}which then enables
us to lift the model-checking algorithm of \cite{grokresie14} to an
algorithm for the counting problem (Section~\ref{sec:nd}).

The rest of the paper is structured as follows.
Section~\ref{section:BasicNotation} provides basic notations, 
Section~\ref{section:FOC-Definition} recalls the definition of $\FOCP$
of \cite{KS_LICS17},
Section~\ref{section:FOC-QueryEval} provides the hardness results for
$\FOCP$ on unranked trees and strings with a linear order,
Section~\ref{sec:FOunC} introduces $\FOunCP$ and gives a precise
formulation of our main result, and
Section~\ref{sec:conclusion} points out directions for future work.
\enlargethispage{\baselineskip}

\NewText{\paragraph{Acknowledgements}
We thank Charlotte Lenz for bringing to our
attention that Section~7 of this paper's previous versions
\cite{GSarxiv2017,GS2018} contained a serious flaw in the
proof of \cite[Theorem~7.1]{GSarxiv2017,GS2018}.

\paragraph{Changelog}
This paper is the new, corrected version of
\cite{GSarxiv2017}. Compared to the previous version, we have
significantly changed the content of Section~\ref{subsec:rploc}: instead of the
previously used notion of ``formulas having $q$-rank at most $\ell$'' of
\cite{grokresie14,GSarxiv2017,GS2018}, we now use a new notion of
``formulas in $\FOplus[\ell,q]$'' (we could say that such formulas
``have modified $q$-rank at most $\ell$''). And we use a variant of 
a ``rank-preserving Gaifman normal form theorem'' of \cite{GS2026},
which enables us to decompose formulas in $\FOplus[\ell,q]$ into a
Boolean combination of 
sentences and of $r$-local formulas in $\FOplus[\ell,q]$, where $r$ is
a number that only depends on $\ell$ and $q$.
This allows us to formulate a corrected
version of \cite[Theorem~7.1]{GSarxiv2017,GS2018}, which is now named
Corollary~\ref{cor:locality2}. We formulated this result in such a
way that only minor changes were necessary in the remaining parts of
Section~\ref{sec:Locality} and in Section~\ref{sec:nd} in order to
obtain correct statements and proofs.\footnote{\NewText{We are aware of the fact that the proof can further be simplified considerably by utilising
  that the formulas $\psi^i_{G,I}$ provided by
  Corollary~\ref{cor:locality2} are \emph{$r$-local}: this, in prinicple, allows
  to completely drop the notion of \emph{cover terms}. We plan to
  implement this further simplification in the upcoming journal
  version of the paper.
}}}

\section{Basic notation}\label{section:BasicNotation}

We write 
$\bZ$, $\bN$, and $\bN_{\ge1}$ for the sets of integers, non-negative integers, and
positive integers, resp.
For all $m,n\in\bN$, we write $[m,n]$ for the set
$\setc{k\in\bN}{m\le k\le n}$, and we let $[m]=[1,m]$.
For a $k$-tuple $\ov{x}=(x_1,\ldots,x_k)$ we write $|\ov{x}|$ to
denote its \emph{arity} $k$.
By $\emptytuple$ we denote the empty tuple, i.e., the tuple of arity 0.

A \emph{signature} $\sigma$ is a finite set of relation symbols.
Associated with every relation symbol $R\in\sigma$ is a
non-negative integer $\ar(R)$ called the \emph{arity} of $R$.  The
\emph{size} $\size{\sigma}$ of a signature $\sigma$ is 
the sum of the arities of its relation symbols.  
A \emph{$\sigma$-structure} $\cA$ consists of a finite
non-empty
set $A$ called the \emph{universe} of $\cA$, and a relation
$R^{\cA} \subseteq A^{\ar(R)}$ for each relation symbol $R \in \sigma$.
Note that according to these definitions, all signatures and all
structures considered in this paper are \emph{finite}, signatures are
\emph{relational} (i.e., they do not contain constants or function
symbols), and signatures
may contain relation symbols of arity 0. Note that there are only two 0-ary
relations over a set $A$, namely $\emptyset$ and
$\set{\emptytuple}$.

We write $\cA\cong\cB$
to indicate
that two $\sigma$-structures $\cA$ and $\cB$ are isomorphic.
A $\sigma$-structure $\B$ is the \emph{disjoint union} of two
$\sigma$-structures $\A_1$ and $\A_2$ if $B=A_1\cup A_2$, $A_1\cap
A_2=\emptyset$, and $R^{\B}=R^{\A_1}\cup R^{\A_2}$ for all
$R\in\sigma$.

Let $\sigma'$ be a signature with $\sigma'\supseteq\sigma$.
A \emph{$\sigma'$-expansion} of a $\sigma$-structure $\A$ is a
$\sigma'$-structure $\B$ such that $B=A$ and $R^{\B}=R^{\A}$ for every
$R\in\sigma$.
Conversely, if $\B$ is a $\sigma'$-expansion of 
$\A$, then $\A$ is called the
\emph{$\sigma$-reduct} of $\B$.

A \emph{substructure} of a $\sigma$-structure $\A$ is a
$\sigma$-structure $\B$ with universe $B\subseteq A$ and
$R^{\B}\subseteq R^\A$ for all $R\in\sigma$.
For a $\sigma$-structure $\A$ and a non-empty set
$B\subseteq A$, we write $\inducedSubStr{\A}{B}$ to denote the
\emph{induced substructure} of $\A$ on $B$,
i.e., the $\sigma$-structure with universe $B$, where
$R^{\inducedSubStr{\A}{B}} = R^\A \cap B^{\ar(R)}$ for all 
$R\in\sigma$.

Throughout this paper, when speaking of \emph{graphs} we mean
undirected graphs.
The \emph{Gaifman graph} $G_\A$ of a
$\sigma$-structure $\A$ is the graph with vertex set $A$
and an edge between two distinct vertices $a,b\in A$ iff there exists
$R\in\sigma$ and a tuple $(a_1,\ldots,a_{\ar(R)})\in R^{\A}$ such
that $a,b\in\set{a_1,\ldots,a_{\ar(R)}}$.  The structure $\A$ is
called \emph{connected} if its Gaifman graph $G_{\A}$ is connected;
the \emph{connected components} of $\A$ are the connected components
of $G_{\A}$.  

The \emph{distance} $\dist^\A(a,b)$
between two elements $a,b\in A$ is the minimal number of edges of a path from
$a$ to $b$ in $G_\cA$; if no such path
exists, we let $\dist^\cA(a,b)\deff\infty$. 
For a tuple $\ov{a}=(a_1,\ldots,a_k)\in A^k$ and an element $b\in A$ we let
$\dist^{\A}(\ov{a},b)\deff \min_{i\in[k]}\dist(a_i,b)$.
For every $r \ge 0$,
the \emph{$r$-ball of $\ov{a}$ in $\A$} is the set
\ $ N_r^\cA(\ov{a}) \ = \ \setc{b\in A\,}{\,\dist^\cA(\ov{a},b)\le r}$.
The \emph{$r$-neighbourhood of $\ov{a}$ in $\cA$} is defined as 
\ $\Neighb{r}{\cA}{\ov{a}}  \deff  \inducedSubStr{\cA}{\neighb{r}{\cA}{\ov{a}}}$\,.

Let $\VARS$ be a fixed countably infinite set
of \emph{variables}.
A \emph{$\sigma$-interpretation} $\I=(\A,\beta)$ consists of a
$\sigma$-structure $\A$ and an \emph{assignment $\beta$ in $\A$}, i.e.,
$\beta\colon\VARS\to A$.
For $k\in\NN$, for $a_1,\ldots,a_k\in A$,
and for pairwise distinct
$y_1,\ldots,y_k\in\VARS$,
we write
$\beta\frac{a_1,\ldots,a_k}{y_1,\ldots,y_k}$
for the assignment $\beta'$ in $\A$ with $\beta'(y_j)=a_j$ for all
$j\in [k]$, and $\beta'(z)=\beta(z)$ for all
$z\in\VARS\setminus\set{y_1,\ldots,y_k}$.
For $\I=(\A,\beta)$ we let
$\I\frac{a_1,\ldots,a_k}{y_1,\ldots,y_k}
\ = \
\big(\A,\beta\frac{a_1,\ldots,a_k}{y_1,\ldots,y_k}\big)$.

The \emph{order} of a $\sigma$-structure $\cA$ is $|A|$, and the
\emph{size} of $\cA$ is
\ \(
\|\cA\| := |A|+\sum_{R\in\sigma}|R^{\cA}|.
\) \
For a graph $G$ we write $V(G)$ and $E(G)$ to denote its vertex set
and edge set, respectively. Sometimes, we will shortly write $ij$ (or
$ji$) to denote an edge $\set{i,j}$ between the vertices $i$ and $j$.
The \emph{size} of $G$ is $\|G\|:=|V(G)|+|E(G)|$. Note that
up to a constant factor depending on the signature, a structure has
the same size as its Gaifman graph.

\section{Syntax and semantics of $\FOC(\Ps)$}\label{section:FOC-Definition}

In \cite{KS_LICS17}, Kuske and Schweikardt introduced the following
logic $\FOC(\Ps)$ and provided an according notion of Hanf normal form,
which was utilised to obtain efficient algorithms for evaluating
$\FOC(\Ps)$-queries on classes of structures of bounded degree.
The syntax and semantics of $\FOC(\Ps)$ is defined as follows (the
text is taken almost verbatim from \cite{KS_LICS17}).

A \emph{numerical predicate collection} is a triple
$(\Ps,\ar,\sem{.})$ where $\Ps$ is a countable set of
\emph{predicate names}, $\ar\colon\Ps\to\bN_{\ge1}$ assigns the
\emph{arity} to every predicate name, and
$\sem{\P}\subseteq\ZZ^{\ar(\P)}$ is the \emph{semantics} of the
predicate name $\P\in\Ps$.  
Basic examples of numerical predicates are
$\P_{\geq 1}$, 
$\P_=$, $\P_{\leq}$, $\quant{Prime}$
with $\sem{\P_{\geq 1}}\deff\NNpos$,
$\sem{\P_=}\deff\setc{(m,m)}{ m\in\bZ}$, 
$\sem{\P_\leq}\deff\setc{(m,n)\in\bZ^2}{ m\leq n}$, 
$\sem{\quant{Prime}} \deff \setc{n\in\NN}{n\text{ is a prime number}}$.
For the remainder of this paper
let us fix an arbitrary numerical predicate collection
$(\Ps,\ar,\sem{.})$ that contains the predicate $\P_{\geq 1}$.

\begin{definition}[\mbox{$\FOC(\Ps)[\sigma]$}]\label{def:FOC}\label{def:FOCP}
Let $\sigma$ be a signature.
The set of \emph{formulas} and \emph{counting terms} for
\emph{$\FOC(\Ps)[\sigma]$} is built according to the following
rules:
\begin{enumerate}[(1)]
  \item\label{item:atomic} $x_1{=}x_2$ and $R(x_1,\ldots,x_{\ar(R)})$
    are \emph{formulas}, where $R\in\sigma$ and
    $x_1,x_2,\ldots,x_{\ar(R)}$ are variables\footnote{in particular,
      if $\ar(R)=0$, then $R()$ is a formula}
  \item\label{item:bool} if $\varphi$ and $\psi$ are formulas,
    then so are $\lnot\varphi$ and $(\varphi\lor\psi)$
  \item\label{item:exists} if $\varphi$ is a formula and
    $y\in\VARS$, then $\exists y\,\varphi$ is a \emph{formula}
  \item\label{item:Q} if $\P\in\Ps$, $m= \ar(\P)$, and
    $t_1,\ldots,t_m$ are counting terms, then
    $\P(t_1,\ldots,t_m)$ is a \emph{formula}
  \item\label{item:countterm} if $\varphi$ is a formula,
    $k\in\NN$, and $\ov{y}=(y_1,\ldots,y_k)$ is a tuple of $k$ pairwise
    distinct variables, then $\Count{\ov{y}}{\varphi}$ is a
    \emph{counting term}
  \item\label{item:constterm} every integer $i\in\ZZ$ is a
    \emph{counting term}
  \item\label{item:plustimesterm} if $t_1$ and $t_2$ are
    counting terms, then so are $(t_1+t_2)$ and
    $(t_1\cdot t_2)$
\end{enumerate}

Note that first-order logic $\FO[\sigma]$ is the
fragment of $\FOC(\Ps)[\sigma]$ built by using only
the rules \eqref{item:atomic}--\eqref{item:exists}.
Let $\I=(\cA,\beta)$ be a $\sigma$-interpretation. 
For every
formula or counting term $\xi$ of $\FOC(\Ps)[\sigma]$, the semantics
$\sem{\xi}^\I$ is defined as follows.

\begin{enumerate}[(1)]

\item 
 $\sem{x_1{=}x_2}^{\I}=1$ if $a_1{=}a_2$ and $\sem{x_1{=}x_2}^{\I}=0$
 otherwise; \ 
    $\sem{R(x_1,\ldots,x_{\ar(R)})}^{\I}=1$ if
    $(a_1,\ldots,a_{\ar(R)})\in R^\A$, and
    $\sem{R(x_1,\ldots,x_{\ar(R)})}^{\I}=0$ otherwise;
    \\
    where $a_j:=\beta(x_j)$ for $j\in\set{1,\ldots,\max\set{2,\ar(R)}}$

\item $\sem{\nicht\varphi}^{\I}=1-\sem{\phi}^{\I}$ and
    $\sem{(\phi\oder\psi)}^{\I}=
    \max\set{\sem{\phi}^{\I},\sem{\psi}^{\I}}$

\item
    $\sem{\exists
      y\,\varphi}^{\I}=\max\setc{\sem{\phi}^{\I\frac{a}{y}}}{a\in A}$

\item[(4)]
    $\sem{\P(t_1,\ldots,t_m)}^{\I}=1$ if
    $\big(\sem{t_1}^{\I},\ldots,\sem{t_m}^{\I}\big)\in\sem{\P}$, and
    $\sem{\P(t_1,\ldots,t_m)}^{\I}=0$ otherwise

\item[(5)]
    $\sem{\Count{\ov{y}}{\varphi}}^{\I}= \big| \bigsetc{
      (a_1,\ldots,a_k)\in A^k }{
      \sem{\phi}^{\I\frac{a_1,\ldots,a_k}{y_1,\ldots,y_k}}=1 } \big|$,
    \ where $\ov{y}=(y_1,\ldots,y_k)$ 

\item[(6)] $\sem{i}^{\I}=i$

\item[(7)] $\sem{(t_1 + t_2)}^{\I}= \sem{t_1}^{\I} + \sem{t_2}^{\I}$, \
    $\sem{(t_1 \cdot t_2)}^{\I}= \sem{t_1}^{\I} \cdot \sem{t_2}^{\I}$

\end{enumerate}

\end{definition}

By $\FOC(\Ps)$ we denote the union of all $\FOC(\Ps)[\sigma]$ for
arbitrary signatures~$\sigma$.
An \emph{expression} is a formula or a counting term.
 As usual, for a formula $\phi$ and a $\sigma$-interpretation
 $\I$ we will often write 
 $\I\models\phi$ to indicate that $\sem{\phi}^\I =1$. Accordingly,
 $\I\not\models\phi$ indicates that $\sem{\phi}^{\I}=0$. 
If $s$ and $t$ are counting terms, then we write $s-t$ for the
counting term $(s+((-1)\cdot t))$.

\begin{example}\label{example:FOC-formulas}
  The following $\FOC(\Ps)$-formula expresses that
  the sum of the numbers of nodes and edges of a directed graph is a prime:
  \[
     \quant{Prime}\,\big(\,(\,\Count{(x)}{x{=}x}\ + \ \Count{(x,y)}{E(x,y)}\,)\,\big)\,.
  \] 

  The counting term $t\deff\Count{(z)}{E(y,z)}$ denotes the out-degree of
  $y$.

  The $\FOC(\Ps)$-formula $\P_{\geq 1}(t)$ expresses that the
  out-degree of $y$ is $\geq 1$. For better readability of such
  formulas we will often write ``$t\geq 1$'' instead of ``$\P_{\geq
    1}(t)$''.

  The $\FOC(\Ps)$-formula
  \[
    \exists x\ \quant{Prime}\,
    \big(\,\Count{(y)}{\P_=\bigl(\,\Count{(z)}{E(x,z)},\;\Count{(z)}{E(y,z)}\,\bigr)}\,\big)
  \]
  expresses that there is some number $d$ (represented by a
  node $x$ of out-degree $d$) such that the number
  of nodes of out-degree $d$ is a prime.
\end{example}

The construct $\exists z$ binds the variable $z\in\VARS$, and
the construct $\#\ov{y}$ in a counting term binds the 
variables from the tuple~$\ov{y}$; all other occurrences of variables
are free. Formally, the set $\free(\xi)$ of \emph{free variables} of
an $\FOC(\Ps)$-expression
$\xi$ is defined inductively as follows:
\begin{enumerate}[(1)]
 \item
  $\free(x_1{=}x_2)=\set{x_1,x_2}$
  \ and \ 
  $\free(R(x_1,\ldots,x_{\ar(R)}))=\set{x_1,\ldots,x_{\ar(R)}}$
 \item
  $\free(\nicht\phi)=\free(\phi)$ \ and \ $\free((\phi\oder\psi))=\free(\phi)\cup\free(\psi)$
 \item
  $\free(\exists y\,\phi)=\free(\phi)\setminus\set{y}$
 \item
  $\free(\P(t_1,\ldots,t_m))=\free(t_1)\cup\cdots\cup\free(t_m)$
 \item
  $\free(\Count{(y_1,\dots,y_k)}{\phi}) = \free(\phi)\setminus\set{y_1,\ldots,y_k}$
 \item
  $\free(i)=\emptyset$ for $i\in\bZ$
 \item
  $\free((t_1+t_2)) \ = \ \free((t_1\cdot t_2)) \ = \ \free(t_1)\cup\free(t_2)$
\end{enumerate}
We will often write $\xi(\ov{z})$, for
$\ov{z} = (z_1,\ldots,z_n)$ with $n\ge 0$, to indicate that at most
the variables from $\set{z_1,\ldots,z_n}$ are free in the expression~$\xi$.
A \emph{sentence} is a formula without free variables, a
\emph{ground term} is a counting term without free variables.

Consider an $\FOC(\Ps)[\sigma]$-counting term
$t(\ov{x})$, for $\ov{x}=(x_1,\ldots,x_m)$.
For a
$\sigma$-structure $\A$ and a tuple $\ov{a} = (a_1,\ldots,a_m)\in A^m$,
we write
$t^{(\A,\ov{a})}$ or $t^\A[\ov{a}]$ for the integer
$\sem{t}^{(\A,\beta)}$, where $\beta$ is an assignment in $\A$ with
$\beta(x_j)=a_j$ for all $j\in [m]$.
For an $\FOC(\Ps)[\sigma]$-formula
$\phi(\ov{x})$ we write
$(\cA,\ov{a}) \models \varphi$ or
$\cA \models \varphi[\ov{a}]$ to indicate that
$\sem{\phi}^{(\A,\beta)}=1$.
In case that
$m=0$ (i.e., $\phi$ is a sentence and $t$ is a ground term), we
simply write $t^{\A}$ instead of $t^{\A}[\ov{a}]$, and we write
$\A\models \phi$ instead of $\A\models\phi[\ov{a}]$.

Two formulas or two counting terms $\xi$ and $\xi'$ are \emph{equivalent} (for short,
$\xi\equiv\xi'$), if $\sem{\xi}^{\I}=\sem{\xi'}^{\I}$ for every
$\sigma$-interpretation~$\I$.
The size $\size{\xi}$ of an expression is its
length when viewed as a word over the alphabet
$\sigma \cup \VARS \cup \Ps \cup \set{,} \cup
\set{=,\allowbreak \exists,\neg,\lor, (, )}\cup\set{\#,.}$.

\section{The hardness of evaluating $\FOC(\Ps)$-queries}
\label{section:FOC-QueryEval}

In \cite{KS_LICS17} it was shown that on classes of structures of
bounded degree, $\FOC(\Ps)$-query evaluation is fixed-para\-meter
linear (when using oracles for evaluating the numerical predicates in
$\Ps$). In this section, we shall prove that there is no hope of
extending this result to even very simple classes of structures of unbounded
degree such as trees and words: on these classes, the $\FOC(\Ps)$
evaluation problem
is as hard as the evaluation problem for first-order logic on
arbitrary graphs. The latter is known to be PSPACE-complete
\cite{var82} and, in the world of parameterised complexity theory, complete
for the class $\AWstern$ \cite{dowfelreg98} (also see
\cite{flugro06}). The hardness results hold for all $\Ps$ that
contain the ``equality predicate'' $\P_{=}$ or the ``positivity predicate'' $\P_{\geq 1}$.
The $\AWstern$-hardness is the more relevant result for us
here.\footnote{PSPACE-completeness already holds over a
  fixed structure with two elements.} It
shows that the evaluation problem is unlikely to have an algorithm
running in time $f(k)n^c$ for an arbitrary function $f$ and constant
$c$, where $k$ is the size of the input formula and $n$ the size of
the input structure.

To state our result formally, we focus on the model-checking problem,
that is, the query evaluation problem for sentences. The 
\emph{model-checking problem} for a logic $\textup L$ on a class $\Class$ of
structures is the problem of deciding whether a given structure
$\cA\in\Class$ satisfies a given L-sentence $\phi$. A \emph{polynomial
  fpt-reduction} between two such 
problems is a
polynomial time many-one reduction that, given an instance $\cA,\phi$
of the first model-checking problem, computes an instance $\cA',\phi'$
of the second model-checking problem such that $\|\cA'\|$ is
polynomially bounded in $\|\cA\|$ and $\|\phi'\|$ is polynomially
\allowbreak bounded
in $\|\phi\|$. 

\begin{theorem}\label{theo:hardness-trees}
  There is a polynomial fpt-reduction from the model-checking problem
  for $\FO$
  on the class of all graphs to the model-checking problem for
  $\FOC(\{\P_{=}\})$ on the class of all trees.
\end{theorem}

\begin{proof}
  Let $G$ be a graph, and let $\phi$ be an $\FO$-sentence in the
  signature of graphs (consisting of a single binary relation symbol
  $E$). 
  W.l.o.g.\ we assume that $V(G)=[n]$ for some
  $n\ge 1$.
We shall define a tree $T_G$ and an $\FOC(\{\P_{=}\})$-sentence
$\hat\phi$ such that $G$ satisfies $\phi$ if and only if $T_G$
satisfies $\hat\phi$.
We construct the tree $T_G$ as follows. The vertex set $V(T_G)$ consists of
  \begin{itemize}
  \item a ``root'' vertex $r$
  \item a vertex $a(i)$ for every $i\in[n]$
  \item vertices $b_j(i)$ and $c_j(i)$ for every $i\in[n]$ and
    $j\in[i{+}1]$
  \item a vertex $d(i,j)$ for every $i\in[n]$ and every neighbour $j$
    of $i$ in $G$
  \item vertices $e_k(i,j)$ for every $i\in[n]$, every neighbour $j$
    of $i$ in $G$, and every $k\in[j{+}1]$.
  \end{itemize}
  The edge set of $T_G$ consists of
  \begin{itemize}
  \item edges $ra(i)$ for all $i\in[n]$
  \item edges $a(i)b_j(i)$ and $b_j(i)c_j(i)$ for all $i\in[n]$ and
    $j\in[i{+}1]$
  \item edges $a(i)d(i,j)$ and $d(i,j)e_k(i,j)$ or all $i\in[n]$, all
    neighbours $j$ of $i$ in $G$, and
    all $k\in[j{+}1]$.
  \end{itemize}
  Then, $T_G$ is a tree (of height $3$) that can be computed from $G$
  in quadratic time.

  To define $\hat\phi$, we need auxiliary formulas
  $\phi_a(x)$, $\phi_b(x)$, \ldots, $\phi_e(x)$ defining the sets of
  $a,b,\ldots,e$-vertices, respectively. We start from the observations that
  the $c$-vertices $c_j(i)$ are the precisely those vertices of degree $1$
  whose unique neighbour has degree $2$. The $b$ vertices are the
  neighbours of the $c$-vertices, and the $a$-vertices are the neighbours
  of the $b$-vertices that are not $c$-vertices. The root vertex is
  the only vertex adjacent to all $a$ vertices. The $e$-vertices are
  the vertices of degree $1$ that are not $c$-vertices, and the
  $d$-vertices are the neighbours of the $e$-vertices.

  Note that the
  vertices of $G$ 
  are in one-to-one correspondence to the $a$-vertices of $T_G$: vertex $i$
  corresponds to the unique $a$-vertex with exactly $(i{+}1)$
  $b$-neighbours. To express that there is an edge between
  $a$-vertices $x,x'$, we say that $x$ has a $d$-neighbour $y$ such
  that the number of $e$-neighbours of $y$ equals the number of
  $b$-neighbours of $x'$. This is precisely what the following
  $\FOC(\{\P_{=}\})$-formulas says:
  \[
     \psi_E(x,x') \ \ := \ \ \exists y\ \Big(E(x,y)\ \wedge
     \P_{=}\Big(\#z.\big(E(y,z)\wedge\psi_e(z)\big),\#z.\big(E(x',z)\wedge\psi_b(z)\big)\Big)\Big).
  \]
  Now we let $\hat\phi$ be the formula obtained from $\phi$ by
  replacing each atom $E(x,x')$ by $\psi_E(x,x')$ and by relativizing all
  quantifiers to $a$-vertices, that is, replacing subformulas $\exists
  x\psi$ by $\exists x(\psi_a(x)\wedge\psi)$. Clearly, $\hat\phi$
  can be computed from $\phi$ in polynomial time. Moreover, it should
  be clear from the construction that $G$ satisfies $\phi$ if and
  only if $T_G$ satisfies $\hat\phi$. 
\end{proof}

\begin{corollary}\label{cor:hardness-trees}
  The parameterised model-checking problem for $\FOC(\{\P_{=}\})$ on the class of
  all trees is $\textup{AW}[*]$-complete.
\end{corollary}

We encode strings over the alphabet $\Sigma$ as structures $\cA$ of
signature $\sigma:=\{\le\}\cup\{P_a :  a\in\Sigma\}$, where the binary
relation $\le^\cA$ is a linear order of $A$, and the unary relation
$P_a^{\cA}$ consists of the positions of all $a$s in the string, for
each symbol $a\in\Sigma$. 

\begin{theorem}\label{theo:hardness-strings}
  There is a polynomial fpt-reduction from the model-checking problem
  for $\FO$
  on the class of all graphs to the model-checking for
  $\FOC(\{\P_{=}\})$ on the class of all strings of alphabet $\Sigma\deff\set{a,b,c}$. 
\end{theorem}

\begin{proof}
  We use a similar idea as in the proof of
  Theorem~\ref{theo:hardness-trees}. Given a graph $G$ with vertex set
  $[n]$ and an $\FO$-sentence $\phi$, we construct a string $S_G$ and
  an $\FOC(\{\P_{=}\})$-sentence $\hat\phi$ such that $G$ satisfies
  $\phi$ if and only if $S_G$ satisfies $\hat\phi$.

  This time, we use substrings (instead of subtrees) to represent the
  vertices of $G$. For a vertex $i$ with neighbours
  $\{j_1,\ldots,j_m\}$ in $G$, we let $s_i$ be the string
  \[
  ac^ibc^{j_1}bc^{j_2}\ldots bc^{j_m}.
  \]
  Then we let $S_G$ be the concatenation of the $s_i$ for all
  $i\in[n]$. It is easy to complete the proof along the lines of the
  proof of Theorem~\ref{theo:hardness-trees}.
\end{proof}

\begin{corollary}\label{cor:hardness-strings}
  The parameterised model-checking problem for $\FOC(\{\P_{=}\})$ on the class of
  all strings (of alphabet $\Sigma=\set{a,b,c}$) is $\textup{AW}[*]$-complete.
\end{corollary}

\begin{remark}
Note that we can express 
``$\P_=(t_1,t_2)$'' via ``$\nicht\,\P_{\geq 1}(t_1{-}t_2)\;\und\;\nicht \,\P_{\geq 1}(t_2{-}t_1)$''.
Therefore, the results of Theorem~\ref{theo:hardness-trees}, Corollary~\ref{cor:hardness-trees},
Theorem~\ref{theo:hardness-strings}, and Corollary~\ref{cor:hardness-strings} also hold for the 
logics $\FOC(\set{\P_{\geq 1}})$.
\end{remark}

\section{The fragment $\FOunC(\Ps)$ of $\FOC(\Ps)$}\label{section:FOunC}\label{sec:FOunC}

In this section, we define a fragment of $\FOC(\Ps)$
called $\FOunC(\Ps)$.
  This logic is an extension of
$\FO$ that allows to formulate cardinality
conditions concerning terms that have at most one free variable (hence
the subscript 1 in ``$\FOunC$'').
The logic $\FOunC(\Ps)$ is designed in such a way that it, although being
relatively expressive, still allows for efficient
query evaluation algorithms on well-behaved
classes of structures. This paper's main result shows that
$\FOunC(\Ps)$-query evaluation is fixed-parameter tractable
on nowhere dense classes of structures.

\begin{definition}[\mbox{$\FOunC(\Ps)[\sigma]$}]\label{def:FOunC}
Let $\sigma$ be a signature.

The set of \emph{formulas} and \emph{counting terms} of
\emph{$\FOunC(\Ps)[\sigma]$} is built according to the rules
\eqref{item:atomic}--\eqref{item:exists} and \eqref{item:countterm}--\eqref{item:plustimesterm} and
the following restricted version of rule
\eqref{item:Q} of
Definition~\ref{def:FOC}:
\begin{enumerate}[(1')]
\item[(\ref{item:Q}')]
  if $\P\in\Ps$, $m= \ar(\P)$, and 
  $t_1,\ldots,t_m$ are counting terms 
  \emph{such that $\setsize{\free(t_1)\cup\cdots\cup\free(t_m)}\leq 1$},
  then $\P(t_1,\ldots,t_m)$ is a \emph{formula}
\end{enumerate}
\end{definition}

The first two formulas of Example~\ref{example:FOC-formulas} are in $\FOunC(\Ps)$;
the last formula of Example~\ref{example:FOC-formulas} and the formula
$\psi_E(x,x')$ from the proof of Theorem~\ref{theo:hardness-trees} are not.
Based on the logic $\FOunC(\Ps)$, we define the following query language.

\begin{definition}[\mbox{$\FOunC(\Ps)$-queries}]\label{def:FOunC-QL}
Let $\sigma$ be a signature.
An \emph{$\FOunC(\Ps)[\sigma]$-query} is of the form
\begin{equation}\label{eq:query-def}
  \query{(x_1,\ldots,x_k,t_1,\ldots,t_\ell)}{\phi}
  \tag{$*$}
\end{equation}
where $k,\ell\geq 0$, $x_1,\ldots,x_k$ are pairwise distinct
variables, $t_1,\ldots,t_\ell$ are $\FOunC(\Ps)[\sigma]$-counting terms 
with $\free(t_i)\subseteq\set{x_1,\ldots,x_k}$ for each $i\in [\ell]$,
and $\phi$ is an $\FOunC(\Ps)[\sigma]$-formula with
$\free(\phi)=\set{x_1,\ldots,x_k}$.
\\
When evaluated in a $\sigma$-structure $\A$, a query $q$ of the form
\eqref{eq:query-def} returns the result
\ $q(\A) \deff \sem{q}^{\A}\deff$
\[
  \Big\{\ 
    \big(\;a_1,\ldots,a_k,n_1,\ldots,n_\ell \;\big)
  \ : \ 
    \A\models\phi[a_1,\ldots,a_k] \ \text{and} \
    n_j=t_j^{\A}[a_1,\ldots,a_k] \text{ for each } j\in[\ell]
  \ \Big\}\,.
\]
\end{definition}

\smallskip

Let us demonstrate that the usual examples for
uses of the COUNT operation in SQL can be expressed in this
query language.

\begin{example}\label{example:SQL-Count}
In this example we consider $\FOunC(\Ps)$-queries where $\Ps$ is
empty, and deal with the database schema consisting of relations
\textsf{Customer(Id, FirstName, LastName, City, Country, Phone)} and
\textsf{Order(Id, OrderDate, OrderNumber, CustomerId,
  TotalAmount)}.\footnote{taken from
  \url{http://www.dofactory.com/sql/group-by}} 
\medskip

To list the number of customers in each country, one can use the
SQL-statement
\begin{verbatim}
 SELECT Country, COUNT(Id)
 FROM Customer
 GROUP BY Country
\end{verbatim}
or the $\FOunC(\Ps)$-query \
$\query{(x_{co},\,t(x_{co}))}{\phi(x_{co})}$ \
with 
\ $
  \phi(x_{co})
  \ \deff \
  x_{co}{=}x_{co}
$ \ and
\ $t(x_{co})\deff \Count{(x_{id})}{\psi}$ \ with $\psi\deff$
\[
    \exists x_{fi} \,\exists x_{la}\, \exists x_{ci}\,
    \exists x_{ph}\ 
    \texttt{Customer}(x_{id},x_{fi},x_{la},x_{ci},x_{co},x_{ph}).
\]

To return the total number of customers and the total number of orders
stored in the database, we can use the SQL-statement\footnote{This
  statement shall work for MySQL, PostgreSQL, and Microsoft SQL
  server; to make it work for Oracle, the statement has to be appended
  by the line \texttt{FROM dual}.} 
\begin{verbatim}
 SELECT(
         SELECT COUNT(*)
         FROM Customer
       ) AS No_Of_Customers,
       (
         SELECT COUNT(*)
         FROM Order
       ) AS No_Of_Orders
\end{verbatim}
or, equivalently, the $\FOunC(\Ps)$-query
\ $\query{(t_c,\,t_o)}{\phi}$ \
for 
\begin{eqnarray*}
   t_c
 & \deff 
 & \Count{(\ov{x})}{\texttt{Customer}(\ov{x})}
\\
   t_o
 & \deff
 & \Count{(\ov{y})}{\texttt{Order}(\ov{y})}
\end{eqnarray*}
with \,$\ov{x}=(x_{id},x_{fi},x_{la},x_{ci},x_{co},x_{ph})$\, and
\,$\ov{y}=(y_{oid},y_{od},\allowbreak y_{on},y_{cid},y_{ta})$\,
and where $\phi$ is a sentence that is satisfied by \emph{every} database, e.g.,
\ $
  \phi
  \ \deff\
  \nicht\exists z \ \nicht\, z {=} z
$\,.
\medskip

To list the total number of orders for each customer in Berlin, 
we can use the SQL-statement
\begin{verbatim}
 SELECT C.FirstName, C.LastName, COUNT(O.Id)
 FROM Customer C, Order O
 WHERE C.City = Berlin AND O.CustomerId = C.Id 
 GROUP BY C.FirstName, C.LastName
\end{verbatim}
or, equivalently, the
$\FOunC(\Ps)$-query
\[
 \query{\big(x_{fi},\,
   x_{la},\,t(x_{fi},x_{la})\big)}{\phi(x_{fi},x_{la})}
\]
with 
\ $t(x_{fi},x_{la}) \ \deff$
\[
  \# \, (y_{oid}) . \; 
  \exists y_{od} \exists y_{on} \exists y_{ta} \exists x_{id} \exists
  x_{ci} \exists x_{co} \exists x_{ph} \ \Big(
  \texttt{Order}(\ov{y})
  \ \und \ 
  \texttt{Customer}(\ov{x})
  \ \Big)
\]
for $\ov{y}=(y_{oid},y_{od},y_{on},x_{id},y_{ta})$ and
$\ov{x}=(x_{id},x_{fi},x_{la},x_{ci},\allowbreak x_{co},x_{ph})$
and
\
$ \phi(x_{fi},x_{la})\deff$
\[
  \exists x_{id}\exists x_{ci}\exists x_{co} \exists x_{ph} \ \big(\
    \texttt{Customer}(\ov{x}) \;\und\; R_{\textit{Berlin}}(x_{ci})
  \ \big)\,.
\]
Here, we use an atomic statement $R_{\textit{Berlin}}(x_{ci})$ to
express that ``$x_{ci}=\textup{Berlin}$''.
Of course, to avoid such constructions, we could extend the definition
of $\FOunC(\Ps)$ in the usual way by allowing constants
taken from a fixed domain $\textbf{dom}$ of potential database entries
(cf.\ \cite{AHV-Book}).
\end{example}

Our query language is also capable of expressing more complicated queries:

\begin{example}
Consider a numerical predicate collection that contains the equality predicate
$\P_=$ with $\sem{\P_=}=\setc{(m,m)}{m\in\ZZ}$.
For better readability of $\FOunC(\Ps)$ formulas we will write $t=t'$
instead of $\P_=(t,t')$.

Consider the signature $\sigma\deff\set{E,R,B,G}$ where $E$ is a
binary relation symbol and $R$, $B$, $G$ are unary relation
symbols. We view a $\sigma$-structure $\A$ as a directed graph where each
node $a\in A$ may be coloured with 0, 1, 2, or 3 of the colours $R$ (red), $B$
(blue), and $G$ (green).

The ground term
\ $
 t_R
 \deff
 \Count{(x)}{R(x)}
$ \
specifies the total number of red nodes.
The term 
\[
t_\Delta(x)
\quad\deff\quad
\Count{(y,z)}{\big(\,E(x,y)\und E(y,z)\und E(z,x)\,\big)}
\]
specifies the number of directed triangles in which $x$ participates.
The formula
\ $
 \phi_{\Delta,R}(x)
  \;\deff\;
 t_\Delta(x) \,{=}\, t_R
$ \
is satisfied by all nodes $x$ such that the number of triangles in
which $x$ participates is the same as the total number of red nodes.
The ground term
\ $
 t_{\Delta,R}
 \deff
 \Count{(x)}{\phi_{\Delta,R}(x)}
$ \
specifies the total number of 
such nodes.
The term
\[
  t_B(x)
  \quad\deff\quad
  \Count{(y)}{\big(\,E(x,y)\und B(y)\,\big)}
\]
specifies the number of blue neighbours of node $x$.

For the formula
\ $
  \phi_{B,\Delta,R}(x) 
  \; \deff \;
  t_B(x) \,{=}\, t_\Delta(x) \,{+}\, t_{\Delta,R}
$\, \
the 
$\FOunC(\Ps)[\sigma]$-query
\[
 \query{\;(\, x,\ y,\ t_B(x)\cdot
   t_\Delta(y)\,)\ \;}{\ \;\big(\phi_{B,\Delta,R}(x)\,\und\, G(y)\big)\;}
\]
outputs all tuples in $A^2\times \ZZ$ of the form $(x,y,n)$ such that 
$n$ is the product of the number of blue neighbours of $x$ and the number of triangles in which
$y$ participates,
$y$ is green, and 
$x$ is a node whose number of blue neighbours
is equal to the sum of the number of triangles in which $x$
participates and the total number of nodes that participate in exactly
as many triangles as there are red nodes.
\end{example}

When speaking of an \emph{algorithm with $\Ps$-oracle} we mean an algorithm
that has available an oracle to decide, at unit cost, 
whether $(i_1,\ldots,i_m)\in\sem{\P}$ 
when given a $\P\in\Ps$ and a tuple of integers $(i_1,\ldots,i_{m})$ of arity
$m=\ar(\Ps)$.

The paper's main result reads as follows 
(see Section~\ref{sec:nd} for a precise definition of
nowhere dense classes).

\begin{theorem}[Main Theorem]\label{thm:main_FOunC-query-eval}
Let 
$\Class$ be an effectively nowhere dense class of structures.
There is an algorithm with $\Ps$-oracle which receives as input an
$\epsilon>0$, an
$\FOunC(\Ps)$-query $q$ of the form \eqref{eq:query-def}
for some signature $\sigma$, a
$\sigma$-structure $\A$ from $\Class$, and a tuple $\ov{a}\in
A^k$, and decides whether $\A\models\varphi[\ov{a}]$, and if so,
computes the numbers $n_j\deff t_j^{\A}[\ov{a}]$ for all $j\in [\ell]$.
The algorithm's running time is
$f(\size{q},\epsilon)\cdot\size{\A}^{1+\epsilon}$, for a computable function $f$.
\end{theorem}

Since the counting problem for an $\FOunCP$-formula $\phi(\ov{x})$ for $\ov{x}=(x_1,\ldots,x_k)$
coincides with the task of evaluating the ground term
$\Count{\ov{x}}{\phi(\ov{x})}$ of $\FOunCP$, we immediately obtain:

\begin{corollary}\label{cor:CountFO}
On effectively nowhere dense classes $\Class$, the counting problem for
$\FOunCP$ is fixed-parameter almost linear. That is, there is 
an algorithm with $\Ps$-oracle which receives as input an
$\epsilon>0$, an
$\FOunCP$-formula $\phi(\ov{x})$
of some signature $\sigma$, and
$\sigma$-structure $\A$ from $\Class$, and computes the number
$|\phi(\A)|$ of tuples $\ov{a}\in A^{|\ov{x}|}$ with
$\A\models\phi[\ov{a}]$ in time  
$f(\size{\phi},\epsilon)\cdot\size{\A}^{1+\epsilon}$, for a computable function $f$.
\end{corollary}

The first step towards proving Theorem~\ref{thm:main_FOunC-query-eval}
is to use a standard construction for getting rid of the free
variables. Given a query $q$ of the form \eqref{eq:query-def}, we
extend the signature $\sigma$ by fresh unary relation symbols
$X_1,\ldots,X_k$ and let
$\tilde{\sigma}\deff\sigma\cup\set{X_1,\ldots,X_k}$. 
Given a $\sigma$-structure $\A$ and a tuple $\ov{a}\in A^k$, we
consider the $\tilde{\sigma}$-expansion $\tilde{\A}$ of $\A$ where 
$X_i^{\tilde{\A}}\deff\set{a_i}$ for all $i\in[k]$.

It is straightforward to translate $\phi(\ov{x})$ into a
$\tilde{\sigma}$-sentence $\tilde{\phi}$ such that 
$\tilde{\A}\models\tilde{\phi}$ iff
$\A\models\phi[\ov{a}]$; and similarly, for each $j\in[\ell]$ we can
translate the term $t_j(\ov{x})$ into a ground term $\tilde{t}_j$ of
signature $\tilde{\sigma}$ such that
$\tilde{t}_j^{\tilde{\A}}=t_j^\A[\ov{a}]$: \ 
W.l.o.g.\ assume that all occurrences of the variables $x_1,\ldots,x_k$ in $\phi$ and
$t_1,\ldots,t_\ell$ are free.
We can choose
\ $
  \tilde{\phi} \ \deff \
  \exists x_1\,\cdots\,\exists x_k\ \big(\,
    \Und_{i=1}^k X_i(x_i) \ \und \ \phi(\ov{x})
  \,\big)
$\,. For each $j\in[\ell]$, the term $t_j$ is built using $+$ and
$\cdot$ from integers and from terms of the form
$\Count{\ov{y}}{\theta(\ov{x},\ov{y})}$.
By replacing each $\theta(\ov{x},\ov{y})$ by
\ $
  \tilde{\theta}(\ov{y}) \ \deff \
  \exists x_1\,\cdots\,\exists x_k\ \big(\,
    \Und_{i=1}^k X_i(x_i) \ \und \ \theta(\ov{x},\ov{y})
  \,\big)
$\,, \ we obtain a ground term $\tilde{t}_j$ with the desired property.

To prove Theorem~\ref{thm:main_FOunC-query-eval} it therefore
suffices to prove the following.

\begin{lemma}\label{lem:main_FOunC-query-eval}
Let $\Class$ be an effectively nowhere dense class of structures.
There is an algorithm with $\Ps$-oracle which receives as input an
$\epsilon>0$, a 
$\sigma$-structure $\A$ from $\Class$ (for some signature $\sigma$) 
and either an $\FOunC(\Ps)[\sigma]$-sentence $\phi$ or an
$\FOunC(\Ps)[\sigma]$-ground term $t$.
The algorithm decides whether $\A\models\phi$ 
and computes $t^\A$, resp.
Letting $\xi$ be the input expression $\phi$ or $t$, the algorithm's running time is 
$f(\size{\xi},\epsilon)\cdot\size{\A}^{1+\epsilon}$, for a computable function $f$.
\end{lemma}

The remainder of the paper is dedicated to the proof of Lemma~\ref{lem:main_FOunC-query-eval}.
In fact, we prove a slightly stronger result: We cannot only evaluate
sentences and ground terms, but also formulas with one free variable
and unary terms simultaneously at all elements of the input structure,
within the same time bounds.

\section{A decomposition of $\FOunC(\Ps)$}\label{sec:normalform}

The first step towards proving Lemma~\ref{lem:main_FOunC-query-eval}
is to provide a decomposition of $\FOunC(\Ps)$-expressions
into simpler
expressions that can be
evaluated in a structure $\A$ by exploring for each element $a$ in
$\A$'s universe only a local neighbourhood around $a$.
This section's main result is the 
Decomposition Theorem~\ref{thm:normalformFOunC}.

Let us fix a signature $\sigma$.

\subsection{Connected local terms}

The following lemma summarises easy facts concerning neighbourhoods;
the proof is straightforward.

\begin{lemma}\label{lem:basic_facts}
  Let $\A$ be a  $\sigma$-structure,
  $r\geq 0$, $k\geq 1$, and $\ov{a}=(a_1,\ldots,a_k)\in A^k$.
\\
    $\NrA{a_1,a_2}$ is connected $\iff$
    $\dist^\A(a_1,a_2)\leq 2r{+}1$.
\\
    If $\NrA{\ov{a}}$ is connected, then
    $\nrA{\ov{a}}\subseteq \neighb{r+(k-1)(2r+1)}{\A}{a_i}$,
    for each $i\in[k]$.
\end{lemma}

The notion of local formulas is defined as usual \cite{Lib04}:
Let $r\in\NN$.
An $\FOC(\Ps)[\sigma]$-form\-ula $\phi(\ov{x})$ with free variables
$\ov{x}=(x_1,\ldots,x_k)$ is \emph{$r$-local around
  $\ov{x}$}
if for every $\sigma$-structure $\A$ and all
$\ov{a}\in A^k$ we have
\ $\A\models\phi[\ov{a}]
  \iff
 \NrA{\ov{a}}\models\phi[\ov{a}]$\,.
A formula is \emph{local} if it is $r$-local for some $r\in\NN$.

For an $r\in\NN$ it is
straightforward to construct an $\FO[\sigma]$-formula
$\dist^{\sigma}_{\leq r}(x,y)$ such that for every $\sigma$-structure $\A$
and all $a,b\in A$ we have
\[
  \A\models\dist^{\sigma}_{\leq r}[a,b]
  \quad \iff\quad \dist^{\A}(a,b)\leq r\,.
\]
To improve readability, we  write
$\dist^\sigma(x,y)\,{\leq}\, r$ for 
$\dist^{\sigma}_{\leq r}(x,y)$, and 
$\dist^\sigma(x,y)\,{>}\,r$ for $\nicht\dist^{\sigma}_{\leq r}(x,y)$.

For every $k\in\NNpos$ we let $\Graphclass_k$ be the set of all
undirected graphs $G$ with vertex set $[k]$.
For a graph $G \in\Graphclass_k$, a number $r\in\NN$, a 
tuple $\ov{y}=(y_1,\ldots,y_k)$ of $k$ pairwise distinct variables,
we consider the formula
\begin{equation*}\label{eq:delta-formula}
\displaystyle
  \delta^\sigma_{G,r}(\ov{y}) 
  \quad \deff \quad
  \ \
  \Und_{\set{i,j}\in E(G)}\!\!\!\!\! \dist^{\sigma}(y_i,y_j)\,{\leq}\, r
  \ \ \und \!\!\!\!
  \Und_{\set{i,j}\not\in E(G)}\!\!\!\!\dist^{\sigma}(y_i,y_j)\,{>}\,r\,.
\end{equation*}
connected components of the
$r$-neighbourhood $\NrA{\ov{a}}$ correspond to the connected
components of $G$.
Clearly, the formula
$\delta^{\sigma}_{G,2r+1}(\ov{y})$ is $r$-local around its free variables $\ov{y}$.

The main ingredient of our decomposition of $\FOunC(\Ps)$-expressions
are the connected local terms (\emph{cl-terms}, for short),
defined as follows.

\begin{definition}[cl-Terms]\label{def:cl-term}
Let $r\in\NN$ and $k\in\NNpos$.
A \emph{basic cl-term (of radius $r$ and width $k$)} is a ground term
$g$ of the form
 \[
   \Count{(y_1,\ldots,y_k)}{
     \big(\,\psi(y_1,\ldots,y_k)\;\und\;\delta^\sigma_{G,2r+1}(y_1,\ldots,y_k)\,\big) }
 \]
or a unary term $u(y_1)$ of the form
 \[
   \Count{(y_2,\ldots,y_k)}{\big(\,\psi(y_1,\ldots,y_k)\;\und\;\delta^\sigma_{G,2r+1}(y_1,\ldots,y_k)\,\big)}
 \]
 where
 $\ov{y}=(y_1,\ldots,y_k)$ is a tuple of $k$ pairwise distinct variables,
 $\psi(y_1,\ldots,y_k)$ is an $\FO[\sigma]$-formula that is $r$-local around 
 $\ov{y}$, and $G\in\Graphclass_k$ is \emph{connected}.

 A \emph{cl-term (of radius $\leq r$ and width $\leq k$)} is built from basic cl-terms (of
 radius $\leq r$ and width $\leq k$) and integers by
 using rule 
 (\ref{item:plustimesterm}) of Definition~\ref{def:FOC}.
I.e., a cl-term is a polynomial with integer coefficients, built from
basic cl-terms $t_1,\ldots,t_\ell$ (for $\ell\geq 0$).
\end{definition}

\begin{remark}\label{rem:radius}\upshape
Note that cl-terms are ``easy'' with
respect to query evaluation in the following sense.
Consider a basic cl-term $u(y_1)$ of the form
\[
  \Count{(y_2,\ldots,y_k)}{\big(\,\psi(y_1,\ldots,y_k)\und
  \delta^\sigma_{G,2r+1}(y_1,\ldots,y_k)\,\big)}\,.
\] 
Recall from Definition~\ref{def:cl-term} that $G$ is a \emph{connected} graph.
Therefore, given a $\sigma$-structure $\A$ and an element $a_1\in A$,
the number $u^\A[a_1]$ 
can be computed by only considering the $R$-neighbourhood of $a_1$, for 
$R\deff r+(k{-}1)(2r{+}1)$ (cf.\ Lemma~\ref{lem:basic_facts}).
After having computed the numbers $u^{\A}[a_1]$ for all $a_1\in A$,
the ground cl-term $g\deff$
\[
 \Count{(y_1,\ldots,y_k)}{\big(\,\psi(y_1,\ldots,y_k)\und
  \delta^\sigma_{G,2r+1}(y_1,\ldots,y_k)\,\big)}
\] 
can be evaluated easily, since 
$g^\A = \sum_{a_1\in A}u^{\A}[a_1]$.
\end{remark}

Our decomposition of $\FOunC(\Ps)$-expressions
proceeds by induction on the construction of the input expression. 
The main technical tool for the construction is
provided by the following lemma.

\begin{lemma}\label{lem:normalform:terms}
Let $r\geq 0$,
$k\geq 1$, and let $\ov{y}=(y_1,\ldots,y_k)$ be a tuple of
$k$ pairwise distinct variables.
Let $\psi(\ov{y})$ be an  $\FO[\sigma]$-formula 
that is $r$-local around its free variables $\ov{y}$, 
and consider the terms $g$ and $u(y_1)$ with
  \begin{eqnarray*}
     g &\deff &  \Count{(y_1,\ldots,y_k)}{\psi(y_1,\ldots,y_k)}
\\
     u(y_1) & \deff & \Count{(y_2,\ldots,y_k)}{\psi(y_1,\ldots,y_k)}\,.
  \end{eqnarray*}
There exists a ground cl-term $\hat{g}$ and a
unary cl-term $\hat{u}(y_1)$, both of radius $\leq r$ and width $\leq k$,
such that $ \hat{g}^{\A} = g^{\A}$
and $\hat{u}^{\A}[a] = u^{\A}[a]$ 
for every $\sigma$-structure $\A$ and every $a\in A$.

Furthermore, there is an algorithm which upon input of $r$ and
$\psi(\ov{y})$ constructs $\hat{g}$ and $\hat{u}(y_1)$.
\end{lemma}

\begin{proof}
For a $\sigma$-structure $\A$ and a formula $\vartheta(\ov{y})$ we
consider the set
\[
 S_\vartheta^\A \quad \deff\quad
 \setc{\ \ov{a}=(a_1,\ldots,a_k)\in A^k \ }{\ \A\models\vartheta[\ov{a}]\ }\,.
\]
Note that for every graph $G\in\Graphclass_k$ the formula
\[
  \psi_G(\ov{y}) \quad \deff \quad
  \psi(\ov{y})\;\und\; \delta^{\sigma}_{G,2r+1}(\ov{y})
\]
is $r$-local around $\ov{y}$. Furthermore, for every
$\sigma$-structure $\A$, the set $S_\psi^{\A}$ is the disjoint union of
the sets $S_{\psi_G}^{\A}$ for all $G\in \Graphclass_k$.
Therefore, 
\begin{eqnarray*}
  g
& \ \ \equiv \ \
& \sum_{G\in\Graphclass_k}
\Count{(y_1,\ldots,y_k)}{\psi_G(y_1,\ldots,y_k)}\qquad\text{and}
\\
  u(y_1)
& \ \ \equiv \ \
& \sum_{G\in\Graphclass_k}
\Count{(y_2,\ldots,y_k)}{\psi_G(y_1,y_2,\ldots,y_k)}\,.
\\
\end{eqnarray*}
To complete the proof of Lemma~\ref{lem:normalform:terms}, it
therefore suffices to show that for every $G\in\Graphclass_k$ the
terms
\begin{eqnarray*}
  g^\psi_G 
& \ \ \deff \ \ 
& \Count{(y_1,\ldots,y_k)}{\psi_G(y_1,\ldots,y_k)}
  \qquad\text{and}
\\
  u^\psi_G(y_1)
& \ \ \deff \ \ 
& \Count{(y_2,\ldots,y_k)}{\psi_G(y_1,y_2,\ldots,y_k)}
\end{eqnarray*}
are equivalent to cl-terms of radius $r$.
We prove this by an
induction on the number of connected components of $G$. Precisely, we
show that the following statement $(*)_c$ is true for every $c\in\NNpos$.
\begin{enumerate}[$(*)_c$:]
\item[$(*)_c$:]
  For every $k\geq c$, for every tuple $\ov{y}=(y_1,\ldots,y_k)$ of
  $k$ pairwise distinct variables, for every $r\geq 0$, for every
  $\FO[\sigma]$-formula $\psi(\ov{y})$ that is $r$-local around
  $\ov{y}$, and for every graph $G\in\Graphclass_k$ that has at most
  $c$ connected components, the terms $g^\psi_G$ and $u^\psi_G(y_1)$
  are equivalent to cl-terms of radius $r$.
\end{enumerate}
The induction base for $c=1$ is trivial: \ it involves
only \emph{connected} graphs $G$, for which by Definition~\ref{def:cl-term}
the terms $g^\psi_G$ and $u^\psi_G(y_1)$ are basic cl-terms (of radius $r$).

For the induction step from $c$ to $c{+}1$, consider some $k\geq
c{+}1$ and a graph $G=(V,E)\in\Graphclass_k$ that has $c{+}1$ connected
components. 
Let $V'$ be the set of all nodes of $V$ that are connected to the
node $1$, and let $V''\deff V\setminus V'$. 

Let $G'\deff\inducedSubStr{G}{V'}$ and $G''\deff\inducedSubStr{G}{V''}$
be the induced subgraphs of $G$ on $V'$ and $V''$, respectively.
Clearly, $G$ is the disjoint union of $G'$ and $G''$, $G'$ is connected, and $G''$ has $c$ connected components.

To keep notation simple, we assume (without loss of generality)
that $V'=\set{1,\ldots,\ell}$ and $V''=\set{\ell{+}1,\ldots,k}$ for
some $\ell$ with $1\leq \ell<k$.
For any tuple $\ov{z}=(z_1,\ldots,z_k)$ we 
let $\ov{z}{}'\deff (z_1,\ldots,z_\ell)$ and $\ov{z}{}''\deff
(z_{\ell+1},\ldots,z_{k})$.

Now consider a number $r\geq 0$ and the formula
$\delta^{\sigma}_{G,2r+1}(\ov{y})$ for $\ov{y}=(y_1,\ldots,y_k)$.
For every $\sigma$-structure $\A$
and every tuple $\ov{a}=(a_1,\ldots,a_k)\in A^k$ with
$\A\models\delta^{\sigma}_{G,2r+1}[\ov{a}]$,
the $r$-neighbourhood $\NrA{\ov{a}}$ is the disjoint union of the
$r$-neighbourhoods $\NrA{\ov{a}{}'}$ and $\NrA{\ov{a}{}''}$.

Let $\psi(\ov{y})$ be an $\FO[\sigma]$-formula that is $r$-local
around its free variables. By using the Feferman-Vaught Theorem (cf.,
\cite{FefermanVaughtPaper,DBLP:journals/apal/Makowsky04}),
we can compute a decomposition of $\psi(\ov{y})$ into a formula 
$\hat{\psi}(\ov{y})$ (that depends on $G$)
of the form
\[
  \Oder_{i\in I}\ \ \Big(\ 
   \psi_i{}'(\ov{y}{}') \ \und \ 
   \psi_i{}''(\ov{y}{}'')  
  \ \Big)\,,
\]
where $I$ is a finite non-empty set, 
each $\psi_i{}'(\ov{y}{}')$ is an $\FO[\sigma]$-formula that is
$r$-local around $\ov{y}{}'$, 
each $\psi_i{}''(\ov{y}{}'')$ is an $\FO[\sigma]$-formula that is
$r$-local around $\ov{y}{}''$, and 
for every $\sigma$-structure
$\A$ and every $\ov{a}\in A^k$ with
$\A\models\delta^{\sigma}_{G,2r+1}[\ov{a}]$ the following is true:
\begin{enumerate}[(1)]
\item
 there exists at most one $i\in I$ such that \
 $(\A,\ov{a})\ \models \
  \big(\,\psi_i{}'(\ov{y}{}') \und \psi_i{}''(\ov{y}{}'')\,\big)$, \quad and
\item
 $\A\models\psi[\ov{a}] \ \iff \ 
  \A\models\hat{\psi}[\ov{a}]$\,.
\end{enumerate}
This implies that the set $S^{\A}_{\psi_G}$ is the
disjoint union
of the sets
$S^{\A}_{(\psi_i{}'\und\psi_i{}''\und\delta^{\sigma}_{G,r})}$ for all
$i\in I$.
Consequently,
\begin{eqnarray*}
  g_G^{\psi}
&  \quad \equiv \quad
&  \sum_{i\in I}\
\Count{(y_1,\ldots,y_k)}{\big(\,\psi_i{}'(\ov{y}')\,\und\,\psi_i{}''(\ov{y}'')\,\und\,\delta^{\sigma}_{G,2r+1}(\ov{y})\,\big)}
\qquad\text{and} 
\\[1ex]
  u_G^{\psi}(y_1)
&  \quad \equiv \quad
&  \sum_{i\in I}\ 
\Count{(y_2,\ldots,y_k)}{\big(\,\psi_i{}'(\ov{y}')\,\und\,\psi_i{}''(\ov{y}'')\,\und\,\delta^{\sigma}_{G,2r+1}(\ov{y})\,\big)}\,.
\end{eqnarray*}
To complete the proof, it suffices to show that each of the terms
\begin{eqnarray*}
  g_G^{\psi,i}
& \quad\deff\quad
&  
\Count{(y_1,\ldots,y_k)}{\big(\,\psi_i{}'(\ov{y}')\,\und\,\psi_i{}''(\ov{y}'')\,\und\,\delta^{\sigma}_{G,2r+1}(\ov{y})\,\big)}
\qquad\text{and} 
\\[1ex]
  u_G^{\psi,i}(y_1)
&  \quad \deff \quad
& 
\Count{(y_2,\ldots,y_k)}{\big(\,\psi_i{}'(\ov{y}')\,\und\,\psi_i{}''(\ov{y}'')\,\und\,\delta^{\sigma}_{G,2r+1}(\ov{y})\,\big)}
\end{eqnarray*}
is equivalent to a cl-term of radius $r$.

By the definition of the formula $\delta^{\sigma}_{G,2r+1}(\ov{y})$ we
obtain that the formula
\
$\psi_i{}'(\ov{y}')\,\und\,\psi_i{}''(\ov{y}'')\,\und\,\delta^{\sigma}_{G,2r+1}(\ov{y})$
\
is equivalent to the formula
\begin{equation}\label{eq:theta-formulas}
\underbrace{ \Big(\ \psi_i{}'(\ov{y}')\,\und\,\delta^{\sigma}_{G',2r+1}(\ov{y}')\
\Big)}_{\textstyle =:\ \vartheta'(\ov{y}')}
 \ \und \ 
\underbrace{ \Big(\ \psi_i{}''(\ov{y}'')\,\und\,\delta^{\sigma}_{G'',2r+1}(\ov{y}'')\
\Big)}_{\textstyle =:\ \vartheta''(\ov{y}'')}
 \ \und \ 
 \Und_{j'\in V'\atop j''\in V''} \dist^{\sigma}(y_{j'},y_{j''})>2r{+}1 \,.
\end{equation}
Therefore, for every $\sigma$-structure $\A$ we have
\[
 S^{\A}_{\psi'_i\und\psi''_i\und\delta^{\sigma}_{G,r}}
 \quad = \quad
 \big( S^{\A}_{\vartheta'}\times S^{\A}_{\vartheta''} \big) \
 \setminus\ T^{\A}\,,
 \qquad\quad \text{for}
\]
\[
  T^{\A} \quad \deff \quad
  \bigsetc{\ 
   \ov{a}\in A^k \ }{\ 
   \A\models\vartheta'[\ov{a}'], \
   \A\models\vartheta''[\ov{a}''],\ 
   (\A,\ov{a})\not\models  \Und_{j'\in V'\atop j''\in V''}
   \dist^{\sigma}(y_{j'},y_{j''})>2r{+}1
  \ }\,.
\]
Let $\GraphclassH$ be the set of all graphs
$H\in\Graphclass_k$ with $H\neq G$, but
$\inducedSubStr{H}{V'}=G'$ and $\inducedSubStr{H}{V''}=G''$.
Clearly, every $H\in\GraphclassH$
has at most $c$ connected components. Furthermore, 
it is straightforward to see that 
for every $\sigma$-structure $\A$, the set $T^{\A}$ is the disjoint
union of the sets
\[
  T^{\A}_{H}
  \quad\deff\quad
  \bigsetc{\ 
   \ov{a}\in A^k \ \;}{\ \;
   \A\models\vartheta'[\ov{a}'], \ \
   \A\models\vartheta''[\ov{a}''],\ \
   \A\models\delta^{\sigma}_{H,2r+1}[\ov{a}]
  \ }
\]
for all $H\in\GraphclassH$.
Since $T^\A\subseteq (S^\A_{\vartheta'}\times S^\A_{\vartheta''})$, we
obtain that
\[
 \big(g_G^{\psi,i}\big)^{\A}
 \quad = \quad
  |S^\A_{\psi'_i\und\psi''_i\und\delta^{\sigma}_{G,r}}|
  \quad = \quad
  |S^\A_{\vartheta'}|\cdot |S^\A_{\vartheta''}| \ - \ \sum_{H\in
    \GraphclassH} |T^\A_H|\,;
\]
and this holds for every $\sigma$-structure $\A$. Therefore,
\[
  g_G^{\psi,i}
  \quad \equiv \quad
  \underbrace{\Big( \Count{\ov{y}'}{\vartheta'(\ov{y}')}
    \Big)}_{\textstyle =: \ t'}
  \cdot 
  \underbrace{\Big( \Count{\ov{y}''}{\vartheta''(\ov{y}'')}
    \Big)}_{\textstyle =:\ t''}
  \ - \ 
  \sum_{H\in\GraphclassH}
  \underbrace{\Count{\ov{y}}{\big(\;\vartheta'(\ov{y}')\und\vartheta''(\ov{y}'')\und\delta^{\sigma}_{H,2r+1}(\ov{y})
      \;\big)}}_{\textstyle =: t_H}\,.
\]
By the induction hypothesis $(*)_c$, each of the terms $t'$, $t''$, and $t_H$
is equivalent to a cl-term of radius $r$.
Hence, also $g_G^{\psi,i}$ is equivalent to a cl-term of radius $r$.

To complete the proof, we need to show that also
$u_G^{\psi,i}(y_1)$ is equivalent to a cl-term (of radius $r$).
This can be done in a 
We proceed by a similar reasoning as above. Note that for every
$\sigma$-structure $\A$ and every $a_1\in A$,
\[
 \big(u_G^{\psi,i}\big)^{\A}[a_1]
 \quad = \quad
  |S^{\A,a_1}_{\psi'_i\und\psi''_i\und\delta^{\sigma}_{G,2r+1}}|
\]
where $S^{\A,a_1}_{\psi'_i\und\psi''_i\und\delta^{\sigma}_{G,2r+1}}$ is
defined as the set of all tuples $(a_2,\ldots,a_k)\in A^{k-1}$ such
that
\[
    \big(\A,(a_1,a_2,\ldots,a_k)\big)\ \models\
     \psi_i{}'(y_1,y_2,\ldots,y_\ell) \; \und\;
     \psi_i{}''(y_{\ell+1},\ldots,y_k)\;\und \;
     \delta^{\sigma}_{G,2r+1}(y_1,y_2,\ldots,y_k)\,.
\]
By \eqref{eq:theta-formulas} we know that $S^{\A,a_1}_{\psi'_i\und\psi''_i\und\delta^{\sigma}_{G,2r+1}}$ 
is the set of all tuples $(a_2,\ldots,a_k)\in A^{k-1}$ such
that
\[
    \big(\A,(a_1,a_2,\ldots,a_k)\big)\ \; \models\ \;
     \vartheta'(y_1,y_2,\ldots,y_\ell) \ \und\ 
     \vartheta''(y_{\ell+1},\ldots,y_k)\ \und \;
     \Und_{j'\leq \ell \atop j''\geq \ell+1}\dist^{\sigma}(y_{j'},y_{j''})>2r{+}1
     \,.
\]
Analogously as above we have
\[
 S^{\A,a_1}_{\psi'_i\und\psi''_i\und\delta^{\sigma}_{G,2r+1}}
 \quad = \quad
 \big( S^{\A,a_1}_{\vartheta'}\times S^{\A}_{\vartheta''} \big) \
 \setminus\ T^{\A,a_1}\,,
 \qquad\quad \text{where}
\]
\[
 S^{\A,a_1}_{\vartheta'}
 \quad \deff \quad
 \bigsetc{\
  (a_2,\ldots,a_\ell)\in A^{\ell-1}\ \ }{\ \
  \A\models\vartheta'[a_1,a_2,\ldots,a_\ell]\ }
\]
and where $T^{\A,a_1}$ is the set of all tuples
$(a_2,\ldots,a_k)\in A^{k-1}$ such that
\[
   \big(\A,(a_1,a_2,\ldots,a_k)\big)
   \ \ \models\ \ 
   \vartheta'(y_1,y_2,\ldots,y_\ell) \ \und \ 
   \vartheta''(y_{\ell+1},\ldots,y_k) \ \und \  
   \nicht \!\!\! \Und_{j'\leq \ell \atop j''\geq \ell+1}\!\!
   \dist^{\sigma}(y_{j'},y_{j''})>2r{+}1\,.
\]
The set $T^{\A,a_1}$ is the disjoint union of the sets
$T_H^{\A,a_1}$ for all $H\in\GraphclassH$, where
$T_H^{\A,a_1}$ is defined as the set of all tuples
$(a_2,\ldots,a_k)\in A^{k-1}$ for which
\[
    \big(\A,(a_1,a_2,\ldots,a_k)\big)\ \ \models\ \ 
     \vartheta'(y_1,y_2,\ldots,y_\ell) \ \und \ 
     \vartheta''(y_{\ell+1},\ldots,y_k)\ \und \ 
     \delta^{\sigma}_{H,2r+1}(y_1,y_2,\ldots,y_k)\,.
\]
Since $T^{\A,a_1}\subseteq (S^{\A,a_1}_{\vartheta'}\times S^\A_{\vartheta''})$, we
obtain that
\[
 \big(u_G^{\psi,i}\big)^{\A}[a_1]
 \quad = \quad
  |S^{\A,a_1}_{\psi'_i\und\psi''_i\und\delta^{\sigma}_{G,2r+1}}|
  \quad = \quad
  |S^{\A,a_1}_{\vartheta'}|\cdot |S^\A_{\vartheta''}| \ - \ \sum_{H\in
    \GraphclassH} |T^{\A,a_1}_H|\,;
\]
and this holds for every $\sigma$-structure $\A$ and every $a_1\in A$. Therefore,
\[
  u_G^{\psi,i}(y_1)
  \quad \equiv \quad
  t'(y_1)\cdot t'' \ - \ \sum_{H\in\GraphclassH} t_H(y_1)
\]
where
\begin{eqnarray*}
  t_H(y_1) 
& \ \deff \
&
\Count{(y_2,\ldots,y_k)}{\big(\;\vartheta'(y_1,y_2,\ldots,y_\ell)\und\vartheta''(y_{\ell+1},\ldots,y_k)\und\delta^{\sigma}_{H,2r+1}(y_1,y_2,\ldots,y_k)\;\big)}\,,
\\[2ex]
  t'' 
& \ \deff \
& \Count{(y_{\ell+1},\ldots,y_k)}{\vartheta''(y_{\ell+1},\ldots,y_k)}\,,
\\[2ex]
  t'(y_1)
& \ \deff \
& \left\{
   \begin{array}{lll}
     \Count{(y_2,\ldots,y_\ell)}{\vartheta'(y_1,y_2,\ldots,y_\ell)}
    & & \text{if $\ell\geq 2$}\,,
    \\[1ex]
     \Count{(y_2)}{\big(\vartheta'(y_1)\und y_2{=}y_1\big)}
    & & \text{if $\ell=1$}\,.
   \end{array}
  \right.
\end{eqnarray*}
By the induction hypothesis $(*)_c$, each of the terms $t'(y_1)$, $t''$, and $t_H(y_1)$
is equivalent to a cl-term of radius $r$.
Hence, also $u_G^{\psi,i}(y_1)$ is equivalent to a cl-term of radius $r$.
This completes the proof of Lemma~\ref{lem:normalform:terms}.
\end{proof}

As an easy consequence of Lemma~\ref{lem:normalform:terms} we obtain

\begin{lemma}\label{lem:normalform:terms2}
 Let $s\geq 0$ and let
  $\chi_1, \ldots, \chi_s$ be arbitrary
  sentences of signature $\sigma$.\footnote{We do not restrict
    attention to $\FO[\sigma]$-sentences here --- the $\chi_j$'s may
    be sentences of any logic, e.g., $\FOC(\Ps)[\sigma]$.} 
  Let $r\geq 0$, $k\geq 1$, and let 
  $\ov{y}=(y_1,\ldots,y_k)$ be a tuple of $k$ pairwise
  distinct variables. Let
  $\phi(\ov{y})$ be a Boolean combination of the
  sentences $\chi_1,\ldots,\chi_s$ and
  of $\FO[\sigma]$-formulas that are $r$-local around their free
  variables $\ov{y}$. 
  Consider the ground term
  \[
     g \ \ \deff \ \ \Count{(y_1,\ldots,y_k)}{\phi(y_1,\ldots,y_k)}
  \]
  and the unary term
  \[
     u(y_1) \ \ \deff \ \ \Count{(y_2,\ldots,y_k)}{\phi(y_1,y_2,\ldots,y_k)}\,.
  \]
  For every $J\subseteq [s]$ there is a ground cl-term $\hat{g}_J$ and
  a unary cl-term $\hat{u}_J(y_1)$ (both of radius $\leq r$ and width
  $\leq k$) such
  that for every $\sigma$-structure
  $\A$ there is exactly one set $J\subseteq [s]$ such that
  \[
    \A \ \ \models \ \
    \chi_J \ \ \deff \ \ \Und_{j\in J}\chi_j \ \und \Und_{j\in
      [s]\setminus J} \nicht\,\chi_j\,,
  \]
  and for this set $J$ we have
  \  $\hat{g}_J^\A = g^\A$ \ and \
    $\hat{u}_J^\A[a] = u^\A[a]$ \ for every $a\in A$.

  Furthermore, there is an algorithm
  which upon input of $r$, $\phi(\ov{y})$, and $J$ constructs $\hat{g}_J$ and
  $\hat{u}_{J}(y_1)$. 
\end{lemma}
\begin{proof}
We can assumme w.l.o.g.\ that $\phi(\ov{y})$ is of the form
\[
  \Oder_{J\subseteq [s]} \big(\;
   \chi_J \ \und \ \psi_J(\ov{y})
  \;\big)
\]
where, for each $J\subseteq[s]$, $\psi_J(\ov{y})$ is an
$\FO[\sigma]$-formula that is $r$-local around its free variables
$\ov{y}$.

For every $J\subseteq [s]$ let $\hat{g}_J$ and $\hat{u}_J(y_1)$ be the
cl-terms obtained by Lemma~\ref{lem:normalform:terms} for the terms
$g_J\deff\Count{\ov{y}}{\psi_J(\ov{y})}$ and
$u_J(y_1)\deff\Count{(y_2,\ldots,y_k)}{\psi_J(\ov{y})}$.

Now consider an arbitrary $J\subseteq [s]$ and a $\sigma$-structure
$\A$ with $\A\models\chi_J$. Clearly, 
\[
\begin{array}{rcccll}
    g^\A 
 &  = 
 &  \big(\Count{\ov{y}}{\psi_J(\ov{y}})\big)^\A
 &  = 
 &  \hat{g}_J^\A\,,
 &
\end{array}
\] 
and
\[
\begin{array}{rcccll}
    u^\A[a] 
 &  = 
 &  \big(\Count{(y_2,\ldots,y_k)}{\psi_J(\ov{y}})\big)^\A[a]
 &  = 
 &  \hat{u}_J^\A[a]\,,
 &  \text{\quad for every $a\in A$.}
\end{array}
\] 
Hence, the proof of Lemma~\ref{lem:normalform:terms2} is complete.
\end{proof}

\subsection{A connected local normalform for \mbox{$\FO$}}

\begin{definition}
A formula in \emph{Gaifman normal form} is a 
Boolean combination of
$\FO[\sigma]$-formulas $\psi(\ov{x})$ that are local around their free
variables $\ov{x}$, and of
\emph{basic local sentences}, i.e., 
$\FO[\sigma]$-sentences $\chi$ of the form
\[
  \exists y_1\cdots \exists y_k\;\Big(\;
   \!\!\Und_{1\leq i<j\leq k}\!\!\dist^{\sigma}(y_i,y_j)>2r 
   \ \und \ 
   \!\!\Und_{1\leq i\leq k}\psi(y_i)\!\!
  \;\Big)\,,
\]
where $k\geq 1$, $r\geq 0$, and $\psi(y)$ is an $\FO[\sigma]$-formula that is $r$-local around its
unique free variable $y$. The number $r$ is called the \emph{radius}
of $\chi$.
\end{definition}

\begin{theorem}[Gaifman \cite{GaifmanPaper}]\label{thm:Gaifman}
Every $\FO[\sigma]$-formula $\phi(\ov{x})$ is equivalent to a 
formula in Gaifman normal form.

Furthermore, there is an algorithm which transforms an input formula $\phi(\ov{x})$
into an equivalent formula $\phi'(\ov{x})$ in Gaifman normal form. The algorithm also
outputs the radius of each basic local sentence of $\phi'$, and a
number $r$ such that every local formula $\psi(\ov{x})$ in
$\phi'$ is $r$-local around $\ov{x}$.
\end{theorem}

By combining Lemma~\ref{lem:normalform:terms} with Gaifman's locality
theorem, we obtain the following normal form for
$\FO$, which may be of independent interest.

\begin{theorem}[cl-Normalform]
\label{cor:normalformFO}
Every $\FO[\sigma]$-formu\-la $\phi(\ov{x})$ is equivalent to a Boolean
combination of
 $\FO[\sigma]$-formulas $\psi(\ov{x})$ that are local
 around their free variables $\ov{x}$, and of
 statements of the form ``$g\geq 1$'', for a ground cl-term $g$.

Furthermore, there is an algorithm which transforms an input
$\FO[\sigma]$-formula $\phi(\ov{x})$ into an equivalent such formula
$\phi'(\ov{x})$. The algorithm also outputs the radius of
each ground cl-term in $\phi'$, and a number $r$ such that every local
formula $\psi(\ov{x})$ in $\phi'$ is $r$-local around $\ov{x}$.
\end{theorem}

\begin{proof}
By Theorem~\ref{thm:Gaifman} it suffices to translate a basic local
sentence into a statement of the form ``$g\geq 1$'' for a ground
cl-term $g$. 

For a basic local sentence
\ $
 \chi \deff 
  \exists y_1\cdots \exists y_k\;
  \vartheta(y_1,\ldots,y_k)
$ \ with \
$  \vartheta(y_1,\ldots,y_k) \deff$
\[
   \Und_{1\leq i<j\leq k}\dist^{\sigma}(y_i,y_j)>2r 
   \ \ \und \ \ 
   \Und_{1\leq i\leq k}\psi(y_i)
\]
let $g_\chi$ be the ground term
\[
  g_\chi \quad \deff \quad
  \Count{(y_1,\ldots,y_k)}{\vartheta(y_1,\ldots,y_k)}\,.
\]
Note that $\vartheta(y_1,\ldots,y_k)$ is $r$-local around its free
variables.
Hence, by Lemma~\ref{lem:normalform:terms} we obtain a ground cl-term
$\hat{g}_\chi$ such that $\hat{g}_{\chi}^{\A}=g_{\chi}^{\A}$ for every
$\sigma$-structure $\A$.
Furthermore, 
\ $
  \A\models\chi
   \iff
  g_{\chi}^{\A} \geq 1
   \iff
  \hat{g}_{\chi}^{\A}\geq 1\,.
$
This completes the proof of Theorem~\ref{cor:normalformFO}.
\end{proof}
 
We use the notion \emph{cl-normalform} to denote the 
formulas $\phi'(\ov{x})$
provided by Theorem~\ref{cor:normalformFO}.
Note that these cl-normalforms do
not necessarily belong to $\FO$, but can be viewed as 
formulas in $\FOunC(\set{\P_{\geq 1}})$ (recall that $\sem{\P_{\geq 1}}=\NNpos$), since statements of the form
``$g\geq 1$'' can be expressed via $\P_{\geq 1}(g)$.

\subsection{A decomposition of $\FOunC(\Ps)$-expressions}

Our decomposition of $\FOunC(\Ps)$ utilises
Theorem~\ref{cor:normalformFO} and is
based on an induction on the maximal nesting depth of constructs of
the form $\#\ov{y}$.
We call this nesting depth the 
\emph{$\#$-depth} $\countr(\xi)$ of a given formula or term
$\xi$. Formally, $\countr(\phi) $ is defined as follows:
\begin{enumerate}[(1)]
\item
  $\countr(\phi) \ \deff \ 0$, \ if $\phi$ is a formula of the form \
  $x_{1}{=}x_2$ \ or \
  $R(x_1,\ldots,x_{\ar(R)})$
\item
  $\countr(\nicht\phi)\ \deff \ \countr(\phi)$
  \ and \
  $\countr((\phi\oder\psi))\ \deff \ \max\set{\countr(\phi),\countr(\psi)}$
\item
  $\countr(\exists y\,\phi) \ \deff \ \countr(\phi)$
\item
  $\countr(\P(t_1,\ldots,t_m)) \ \deff \ \max\set{\countr(t_1),\ldots,\countr(t_m)}$,
\item
  $\countr(\Count{\ov{y}}{\phi}) \ \deff \ \countr(\phi)+1$,
\item
  $\countr(i)\ \deff \ 0$, \ for all terms $i\in\bZ$
\item
  $\countr((t_1+t_2)) \ \deff \ \countr((t_1\cdot t_2))
  \ \deff \ \max\set{\countr(t_1),\countr(t_2)}$, \ for all terms $t_1$ and $t_2$.
\end{enumerate}

The base case of our decomposition of $\FOunC(\Ps)$ is
provided by the following lemma; the lemma's proof utilises 
Theorem~\ref{cor:normalformFO}.

\begin{lemma}\label{lem:normalform:countrOne}
Let $\phi$ be an $\FOunC(\Ps)[\sigma]$-formula of the form
$\P(t_1,\ldots,t_m)$ with $\P\in \Ps$,
$m=\ar(\P)$, and where $t_1,\ldots,t_m$ are counting terms of
$\#$-depth at most $1$.
Then, $\phi$ is equivalent to a Boolean combination of
\begin{enumerate}[(i)]
\item
 formulas of the form $\P(t'_1,\ldots,t'_{m})$, for cl-terms
 $t'_1,\ldots,t'_{m}$ where $\free(t'_i)=\free(t_i)$ for all
 $i\in[m]$, 
\item
 statements of the form ``$g\geq 1$'' for ground cl-terms $g$, and
\item
 statements of the form $\P'(i_1,\ldots,i_{m'})$ for $\P'\in\Ps$,
 $m'=\ar(\P')$, and integers $i_1,\ldots,i_{m'}$.
\end{enumerate}
Furthermore, there is an algorithm which transforms an input formula
$\phi$ into such a Boolean combination $\phi'$, and which
also outputs the radius of each cl-term in $\phi'$.
\end{lemma}
\begin{proof}
Let $\phi$ be  of the form 
$\P(t_1,\ldots,t_m)$ with $\P\in \Ps$,
$m=\ar(\P)$, and where $t_1,\ldots,t_m$ are counting terms of
$\#$-depth at most $1$.
From Definition~\ref{def:FOunC} we know that either
$\free(\phi)=\emptyset$ or $\free(\phi)=\set{x}$ for a variable $x$.
Furthermore, we know that
for every $i\in [m]$ the
counting term $t_i$ is built by using addition and multiplication
based on integers and on
counting terms $\theta'$ of the form~$\Count{\ov{z}}{\theta}$, 
for a tuple of variables $\ov{z}=(z_1,\ldots,z_k)$,
such that $\free(\theta)\setminus\set{z_1,\ldots,z_k}\subseteq\set{x}$.
Let $\Theta'$ be the set of all these counting terms $\theta'$ and let
$\Theta$ be the set of all the according formulas $\theta$.  

By assumption we have $\countr(\phi)\leq 1$. Therefore, every
$\theta\in\Theta$ has $\#$-depth 0.
We can thus view each such $\theta$ as an
$\FO[\sigma]$-formula, possibly enriched by atomic sentences of the form
$\P'(i_1,\ldots,i_{m'})$ with $\P'\in\Ps$, $m'=\ar(\P')$, and integers
$i_1,\ldots,i_{m'}$.

By Theorem~\ref{cor:normalformFO}, for each $\theta$ in $\Theta$ we
obtain an equivalent formula $\phi^{(\theta)}$ in cl-normalform,
possibly enriched by atomic sentences of the form
$\P'(i_1,\ldots,i_{m'})$ with $\P'\in\Ps$, $m'=\ar(\P')$, and integers
$i_1,\ldots,i_{m'}$.
Let $\Phi$ be the set of all these $\phi^{(\theta)}$.  

For each $\theta$ in $\Theta$, the formula
$\phi^{(\theta)}$ is a Boolean
combination of (a) $\FO[\sigma]$-formulas that are local around the free variables
of $\theta$, and (b) statements of
the form ``$g\geq 1$'' for a ground cl-term $g$, and (c) statements of
the form $\P'(i_1,\ldots,i_{m'})$ with $\P'\in\Ps$, $m'=\ar(\P')$, and integers
$i_1,\ldots,i_{m'}$.

Let $\chi_1,\ldots,\chi_s$ be a
list of all statements of the forms (b) or (c), such that each formula
in $\Phi$ is a Boolean combination of statements in
 $\set{\chi_1,\ldots,\chi_s}$ and of $\FO[\sigma]$-formulas that are
 local around their free variables.
For every $J\subseteq [s]$ let \ $\chi_J \;:=\; \Und_{j\in J}\chi_j
\und \Und_{j\in [s]\setminus J}\nicht\chi_j$.

Let $r\in\NN$ be such that each of the local $\FO[\sigma]$-formulas that occur in
a formula in $\Phi$ is $r$-local around its free variables.
For each $\theta'$ in $\Theta'$ of the form $\Count{\ov{z}}{\theta}$, we
apply Lemma~\ref{lem:normalform:terms2} to the term
\[
  t^{(\theta')}
  \ \ \deff \ \ 
  \Count{\ov{z}}{\phi^{(\theta)}}
\]
and obtain for every $J\subseteq [s]$ a cl-term
$\hat{t}^{(\theta')}_J$
with the same free variables as $\theta'$,
for which the following is true:

\begin{itemize}
\item 
If $\free(\theta')=\emptyset$, then
\ $(\theta')^{\A} = (\hat{t}^{(\theta')}_J)^{\A}$ \ 
for every $\sigma$-structure $\A$ with $\A\models\chi_J$.

\item 
If $\free(\theta')=\set{x}$, then
\ $(\theta')^{\A}[a] = (\hat{t}^{(\theta')}_J)^{\A}[a]$ \ 
for every $\sigma$-structure $\A$ with $\A\models\chi_J$ and every
$a\in A$.
\end{itemize}

\noindent
Thus, for each $J\subseteq [s]$ we have
\[
  \big(\ \chi_J \ \und \ \P(t_1,\ldots,t_m) \ \big)
  \ \ \equiv\ \
  \big(\ \chi_J \ \und \ \P(t_{1,J},\ldots,t_{m,J})\ \big)
\]
where, for every $i\in [m]$, we let $t_{i,J}$ be the
cl-term obtained from $t_i$ by replacing
each occurrence of a term $\theta'\in\Theta'$ by the
term $\hat{t}^{(\theta')}_J$.  
In summary, we obtain the following:
\begin{align*}
  \phi
   & \ \ = \ \  \P(t_1,\ldots,t_m)\\
   & \ \  \equiv \ \
   \displaystyle\Oder_{J\subseteq [s]} \big(\
    \chi_J \ \und \ \P(t_1,\ldots,t_m) 
  \ \big)\\
   & \ \ \equiv \ \
  \displaystyle\Oder_{J\subseteq [s]} \big(\
    \chi_J \ \und \ \P(t_{1,J},\ldots,t_{m,J})
  \ \big)
 \quad =: \ \
 \phi'\,.
\end{align*}
The formula $\chi_J$ is a Boolean combination of statements of the
form ``$g\geq 1$'' for ground cl-terms $g$ and statements of the form
$\P'(i_1,\ldots,i_{m'})$ for $\P'\in\Ps$, $m'=\ar(\P')$, and integers $i_1,\ldots,i_{m'}$.
Furthermore, each of the terms $t_{i,J}$ is a cl-term with $\free(t_{i,J})=\free(t_i)$.
Thus, the proof of Lemma~\ref{lem:normalform:countrOne} is complete.
\end{proof}

\begin{theorem}[Decomposition of $\FOunC(\Ps)$]\label{thm:normalformFOunC}
Let $z$ be a fixed variable in $\VARS$. 
For every $d\in\NN$ and every $\FOunC(\Ps)[\sigma]$-expression $\xi$
which is either a formula
$\phi(\ov{x})$ or a ground term $t$ of $\#$-depth $\countr(\xi)=d$, there exists a sequence 
$(L_1\ldots,L_{d+1},\xi')$ with the following properties.
\begin{enumerate}[(I)]
\item
$L_i=(\tau_i,\iota_i)$, for every $i\in\set{1,\ldots,d{+}1}$, where
\begin{enumerate}[$\bullet$]
\item
 $\tau_i$ is a finite set of relation symbols of arity $\leq 1$ that do not belong
 to $\sigma_{i-1}\deff \sigma\cup\bigcup_{j<i}\tau_j$, and
\item
 $\iota_i$ is a mapping that associates with every symbol $R\in\tau_i$
 a formula $\iota_i(R)$  
 \begin{enumerate}[(i)]
  \item
   of the form
   $\P(t_1,\ldots,t_{m})$, where $\P\in\Ps$, $m=\ar(\P)$, and
   $t_1,\ldots,t_m$ are cl-terms
   of signature $\sigma_{i-1}$,
   such that $\free(t_j)\subseteq\set{z}$ for each $j\in[m]$, or
  \item
   of the form ``$g\geq 1$'' for ground cl-terms $g$ of
   signature $\sigma_{i-1}$.
  \end{enumerate}
  If $R$ has arity 0, then $\iota_i(R)$ has no free variable.
  If $R$ has arity 1, then $z$ is the unique free variable of
   $\iota_i(R)$ (thus, $\iota_i(R)$ is of the form (i)).
\end{enumerate}
\item
If $\xi$ is a ground term $t$, then 
$\xi'\deff t'$ is a ground cl-term of signature $\sigma_{d+1}$.

If $\xi$ is a formula $\phi(\ov{x})$, then
$\xi'\deff \phi'(\ov{x})$ is a Boolean combination of 
     \textup{(A)}~$\FO[\sigma_{d+1}]$-formulas $\psi(\ov{x})$ that are local around
     their free variables $\ov{x}$, where
     $\sigma_{d+1}\deff\sigma\cup\bigcup_{1\leq i\leq d{+}1}\tau_i$, and
     \textup{(B)}~statements of the form $R()$ where $R$ is a 0-ary relation
     symbol in $\sigma_{d+1}$.
In case that $\free(\phi)=\emptyset$, 
$\phi'$ only contains statements of 
the latter form.

\item
For every $\sigma$-interpretation $\I=(\A,\beta)$ we have
\[
  \sem{\xi}^{\I}
  \quad = \quad
  \sem{\xi'}^{\I_{d+1}}
\] 
(i.e.,
$t^\A=(t')^{\A_{d+1}}$ in case that $\xi$ is a ground term $t$, and
$\I\models\phi$ iff $\I_{d+1}\models\phi'$ in case that $\xi$ is a
formula $\phi$),
where $\I_{d+1}=(\A_{d+1},\beta)$, and $\A_{d+1}$ is the $\sigma_{d+1}$-expansion of $\A$ defined as follows: \
$\A_0\deff \A$, and for every $i\in[d{+}1]$, $\A_i$ is the
$\sigma_i$-expansion of $\A_{i-1}$, where
for every unary $R\in\tau_i$ we have
\[
  R^{\A_i} \ \ \deff \ \
     \setc{\ a\in A\ }{\ (\A_{i-1},a)\;\models\; \iota_i(R)\ }
\]
and for every 0-ary $R\in\tau_i$ we have
\[
  R^{\A_i} \ \ \deff \ \ 
  \left\{
   \begin{array}{cll}
     \set{\,\emptytuple\,} & & \text{if } \A_{i-1}\models\iota_i(R)\,,
   \\[1ex]
     \emptyset & & \text{if } \A_{i-1}\not\models\iota_i(R)\,.
   \end{array}
  \right.
\]
\end{enumerate}
Moreover, there is an algorithm which constructs such a sequence
$D=(L_1,\ldots,L_{d+1},\xi')$ for an input expression $\xi$.
The algorithm also outputs the radius of each cl-term in $D$, and a
number $r$ such that every formula of type \textup{(A)} in $\phi'$ is $r$-local
around its free variables.
\end{theorem}

\begin{proof}
We first prove the theorem's statement for the case that the input
expression $\xi$ is a formula $\phi(\ov{x})$.

We proceed by induction on $i$ to construct for all
$i\in\set{0,1,\ldots,d}$
a tuple $L_i=(\tau_i,\iota_i)$ and an $\FOunC[\sigma_i]$-formula
$\phi_i(\ov{x})$ of $\#$-depth $(d{-}i)$, such that
for every $\sigma$-interpretation $\I=(\A,\beta)$ and the
interpretation $\I_i\deff(\A_i,\beta)$
we have \ $\I\models\phi \iff \I_i\models\phi_i$.

For $i=0$ we are done by letting $\tau_0\deff\emptyset$, $\sigma_0\deff\sigma$, $\phi_0\deff\phi$,
and letting $\iota_0$ be the mapping with empty domain.
Now assume that for some $i<d$, we have already constructed
$L_{i}=(\tau_{i},\iota_{i})$ and 
$\phi_{i}$.
To construct $L_{i+1}=(\tau_{i+1},\iota_{i+1})$ and $\phi_{i+1}$, we proceed as follows. 

Let $\Pi$ be the set of all $\FOunC(\Ps)[\sigma_{i}]$-formulas of
$\#$-depth $\leq 1$ of the form $\P(t_1,\ldots,t_m)$, for $\P\in\Ps$ and $m=\ar(\P)$, 
that occur in $\phi_{i}$.

Now consider an arbitrary formula $\pi$ in $\Pi$ of the
form $\P(t_1,\ldots,t_m)$. 
From Definition~\ref{def:FOunC} we know that there is a variable $y$
such that $\free(t_j)\subseteq\set{y}$ for every $j\in[m]$.
By Lemma~\ref{lem:normalform:countrOne}, $\pi$ is
equivalent to a Boolean combination $\pi'$ of 
\begin{enumerate}[(a)]
\item
 formulas of the form $\P(t'_1,\ldots,t'_{m})$, for cl-terms
 $t'_1,\ldots,t'_{m}$ of signature $\sigma_{i}$, where $\free(t'_j)=\free(t_j)\subseteq\set{y}$ for each
 $j\in[m]$, 
\item
 statements of the form ``$g\geq 1$'' for ground cl-terms $g$ of
 signature $\sigma_i$, and
\item
 statements of the form $\P'(i_1,\ldots,i_{m'})$ for $\P'\in\Ps$,
 $m'=\ar(\P')$, and integers $i_1,\ldots,i_{m'}$.
\end{enumerate}
For each statement $\chi$ of the form (b) or (c), we include into
$\tau_{i+1}$ a 0-ary relation symbol $R_\chi$, 
we replace each occurrence of $\chi$
in $\pi'$ with the new atomic formula $R_\chi()$,
and we let
$\iota_{i+1}(R_\chi)\deff \chi$.
For each statement $\chi$ in $\pi$ of the form (a) we proceed as follows.
If $\free(\chi)=\emptyset$, then we include into  
$\tau_{i+1}$ a $0$-ary relation symbol $R_\chi$,
we replace each occurrence of $\chi$
in $\pi'$ with the new atomic formula $R_\chi()$, and
we let $\iota_{i+1}(R_\chi)\deff \chi$.
If $\free(\chi)=\set{y}$, then we include into $\tau_{i+1}$ a unary
relation symbol $R_\chi$, 
we replace each
occurrence of $\chi$ in $\pi'$ by the new atomic formula $R_\chi(y)$,
and we let $\iota_{i+1}(R_\chi)$ be the formula
obtained from $\chi$ by consistently replacing every free occurrence
of the variable $y$ by the variable $z$.
We write $\pi''$ for the resulting formula $\pi'$.

Clearly, $\pi''$ is a quantifier-free $\FO[\sigma_{i+1}]$-formula, for
$\sigma_{i+1}\deff \sigma_i\cup\tau_i$; in particular, it has $\#$-depth
0. It is straightforward to see
that for every $\sigma$-interpretation $\I=(\A,\beta)$ we have
\[
  \I_i\models\pi \quad \iff \quad \I_{i+1}\models\pi''\,,
\]
for $\I_i\deff (\A_i,\beta)$ and $\I_{i+1}\deff (\A_{i+1},\beta)$.

The induction step is completed by letting $\phi_{i+1}$ be
the formula obtained from $\phi_i$ by replacing every occurrence of a
formula $\pi\in\Pi$ by the formula $\pi''$.
It can easily be verified that $\phi_{i+1}$ is an $\FOunC(\Ps)[\sigma_{i+1}]$-formula of
$\#$-depth $\countr(\phi_i)-1 = ((d{-}i){-}1)= (d{-}(i{+}1))$, and that
\ $\I_i\models\phi_i\iff \I_{i+1}\models\phi_{i+1}$.

By the above induction we have constructed $L_1,\ldots,L_d$ and an
$\FOunC(\Ps)[\sigma_d]$-formula $\phi_d$ of $\#$-depth 0.
Since $\countr(\phi_d)=0$, we can view $\phi_d$ as an
$\FO[\sigma_d]$-formula, possibly enriched by atomic sentences of
the form $\P(i_1,\ldots,i_m)$ with $\P\in\Ps$, $m=\ar(\P)$, and
integers $i_1,\ldots,i_m$.
By Theorem~\ref{cor:normalformFO} we obtain an equivalent formula
$\tilde{\phi}$ of signature $\sigma_d$ in cl-normalform, possibly enriched by
atomic sentences of the form $\P(i_1,\ldots,i_m)$ with $\P\in\Ps$,
$m=\ar(\P)$, and integers $i_1,\ldots,i_m$.
I.e., $\tilde{\phi}$ is a Boolean combination of 
\begin{enumerate}[(A)]
 \item $\FO[\sigma_d]$-formulas that are local around their free
   variables $\ov{x}$,
 \item statements of the form ``$g\geq 1$'', for a ground cl-term $g$
   of signature $\sigma_d$, and
 \item statements of the form $\P(i_1,\ldots,i_m)$ with $\P\in\Ps$,
$m=\ar(\P)$, and integers $i_1,\ldots,i_m$.
\end{enumerate}
For each statement $\chi$ of the form (B) or (C) we include into
$\tau_{d+1}$ a new relation symbol $R_\chi$ of arity 0,
we replace each occurrence of $\chi$ in $\tilde{\phi}$ by the new
atomic formula $R_\chi()$, and we let 
$\iota_{d+1}(R_\chi)\deff \chi$.
Letting $\xi'\deff \phi'$ be the resulting formula $\tilde{\phi}$ completes the proof of
Theorem~\ref{thm:normalformFOunC} for the case that the input
expression $\xi$ is a formula.
\medskip

Let us now turn to the case where the input expression $\xi$ is a
ground term $t$. Then, $t$ is built using $+$ and $\cdot$ from
integers and from ground terms $g$ of the form
$\Count{\ov{y}}{\theta(\ov{y})}$.
Let $S$ be the set of all these ground terms $g$, and let
$\Theta$ be the set of all according formulas $\theta(\ov{y})$.
We have already proven the theorem's statement for the case where the
input expression is a formula.
For each $\theta\in\Theta$, we therefore obtain a sequence 
$D_\theta=(L_1^\theta,\ldots,L_{d_\theta+1}^\theta,\theta')$ with
$d_{\theta}=\countr(\theta)$ and
$L_i^\theta=(\tau_i^\theta,\iota_i^\theta)$.

Each term $g\in S$ is of the form $\Count{\ov{y}}{\theta(\ov{y})}$ for
some $\theta\in\Theta$. 
Clearly, $g^\A=(g')^{\A_{d_{\theta}+1}}$ for $g'\deff \Count{\ov{y}}{\theta'(\ov{y})}$.
Note that $\theta'$ is local around its free variables $\ov{y}$.
Therefore, from Lemma~\ref{lem:normalform:terms}(a) we obtain a ground
cl-term $\hat{g}'$ that is equivalent to $g'$.
We let $\xi'\deff t'$ be the ground cl-term obtained from $t$ by
replacing every term $g\in S$ by $\hat{g}'$. We are done by letting 
$L_i\deff(\tau_i,\iota_i)$ for each $i\in\set{1,\ldots,d{+}1}$, where
$\tau_i$ (and $\iota_i$) is the union of $\tau_i^\theta$ (and
$\iota_i^\theta$, respectively) for all $\theta\in \Theta$.
This finally completes the proof of Theorem~\ref{thm:normalformFOunC}.
\end{proof}

We call the sequence $(L_1,\ldots,L_{\countr(\xi)+1},\xi')$ obtained from
Theorem~\ref{thm:normalformFOunC} for an $\FOunC(\Ps)$-formula or
ground term $\xi$
a \emph{cl-decomposition} of $\xi$.

\medskip

Assume, we have available an efficient algorithm $\mathbb{A}$ for computing the
value  $u^\B[b_1]$ of a unary basic cl-term $u(y_1)$ in a structure
$\B$ for all values $b_1\in B$. 
This algorithm can also be used to compute the value of a 
ground basic cl-term
$g\deff\Count{(y_1,\ldots,y_k)}{\psi(y_1,\ldots,y_k)}$ in $\B$, since
$g^\B = \sum_{b_1\in B}u^{\B}[b_1]$ for the unary basic cl-term
$u(y_1) \deff \Count{(y_2,\ldots,y_k)}{\psi(y_1,y_2,\ldots,y_k)}$.

We argue that by Theorem~\ref{thm:normalformFOunC}, the algorithm
$\mathbb{A}$ can also be used to evaluate 
an $\FOunC(\Ps)$-expression $\xi$ that is either a ground
term $t$ or a sentence $\phi$ in a $\sigma$-structure $\A$.
To evaluate $\xi$ in a $\A$ 
we can proceed as follows.
\begin{enumerate}[(1)]
\item
Use Theorem~\ref{thm:normalformFOunC} to compute a cl-decomposition
$D=(L_1,\ldots,L_{d+1},\xi')$ of $\xi$, for $d\deff \countr(\xi)$.
\item
Let $\A_0\deff \A$. 
\item
For each $i\in[d{+}1]$, compute the $\sigma_i$-expansion of $\A_{i-1}$.
To achieve this, consider for each $R\in \tau_i$ the formula
$\iota_i(R)$. This formula is a very simple statement concerning one
or several cl-terms (each of which is a polynomial built from integers
and basic cl-terms). Let $t_1,\ldots,t_s$ be the list of all basic
cl-terms that appear in $\iota_i(R)$.
For each $j\in [s]$ use algorithm $\mathbb{A}$ to compute the values
$t^\A_j[a]$ for all $a\in A$ (resp., the value $t^\A_j$, if $t_j$ is
ground). Then, combine the values and use a $\Ps$-oracle to check for each $a\in A$ whether $\iota_i(R)$ is
satisfied by $(\A_{i-1},a)$, and store the new relation $R^{\A_i}$ accordingly. 
\item
If $\xi$ is a sentence $\phi$, then $\phi'$ is a Boolean combination
of statements of the form $R()$, for 0-ary relation symbols
$R\in\sigma_{d+1}$.
Thus, checking whether $\A_{d+1}\models\phi'$ boils down to evaluating a
propositional formula, and hence is easy.

If $\xi$ is a ground term $t$, then $t'$ is a ground cl-term. I.e.,
$t'$ is a polynomial built from integers and ground basic cl-terms
$t'_1,\ldots,t'_s$ for some $s\geq 1$. 
For each $j\in[s]$ we use algorithm $\mathbb{A}$ 
to compute the value of $t'_j$ in $\A_{d+1}$.
Afterwards, we combine these values to compute the value of $t'$ in
$\A_{d+1}$.
\end{enumerate}

From \cite{fri04,kazseg13,DBLP:conf/icdt/SegoufinV17} we obtain fixed-parameter almost linear
algorithms for counting the number of solutions of $\FO$-queries on
planar graphs, classes of bounded local tree-width,
classes of bounded expansion,
and---most generally---classes of locally bounded expansion.
By the above approach, this immediately provides us with an 
$\FPT$ algorithm for $\FOunCP$ on these classes.
For nowhere dense classes, though, it is not so
easy to generalise the $\FO$ model-checking
algorithm of \cite{grokresie14} to compute the values of unary
cl-terms. The remainder of the paper is dedicated to this task.

\section{Neighbourhood covers and local evaluation}
\label{sec:Locality}

The techniques of the previous section enable us to reduce the
evaluation of $\FOunC(\Ps)$-sentences and ground terms to the 
evaluation of unary basic cl-terms. To obtain an efficient algorithm
for evaluating the latter on structures $\A$ from a nowhere dense
class of structures, we need to provide a variant of basic cl-terms
(along with techniques to decompose such terms)
that are based on so-called neighbourhood covers.

An \emph{$r$-neighbourhood cover} of a structure $\cA$ is a mapping
$\cX:A\to 2^A$ such that for every $a\in A$ the set $\cX(a)$ is
connected in the Gaifman graph $G_\A$ of $\cA$ and it holds that
$N_r^{\cA}(a)\subseteq \cX(a)$. The sets $\cX(a)$ (for $a\in A$), and depending on the
context also the induced
substructures $\cA[\cX(a)]$, are called the \emph{clusters}
of the cover. Usually, we want the clusters to have
small radius, where the \emph{radius} of a connected set
$X\subseteq A$ is the least $s$ such that there is a $c\in X$ such
that $X\subseteq N_s^{\cA[X]}(c)$. Moreover, we want a
neighbourhood cover to be \emph{sparse}, which means that no
$b\in A$ appears in too many of the sets $\cX(a)$. We will see later
(Theorem~\ref{thm:alg-covers})
that in structures from a nowhere dense class of structures we can
efficiently construct sparse $r$-neighbourhood covers of radius at
most $2r$. In this section, we do not have to worry about the radius or
sparsity of neighbourhood covers.

We need some additional terminology
for neighbourhood covers. We write $X\in\cX$ to express
that $X$ is a cluster of $\cX$, i.e., $X=\cX(a)$ for some $a\in A$.
We say that a cluster $X\in\cX$ \emph{$s$-covers} a tuple $\bar a\in A^k$
if $N_s^\A(\bar a)\subseteq X$. Note that $\cX(a)$ $r$-covers $a$, but there
may be other clusters $X\in\cX$ that $r$-cover $a$ as well.

$\FO^+$ is the extension of first-order logic by adding new atomic
formulas $\dist(x,y)\,{\le}\, d$, with the obvious meaning. Note that
$\FO^+$ is only a syntactic extension and not more expressive than
$\FO$, because the ``distance atoms'' $\dist(x,y)\,{\le}\, d$ can be
replaced by first-order formulas. However, a first-order formula
expressing $\dist(x,y)\,{\le}\, d$ has quantifier rank $\log d$.

\subsection{Rank-preserving locality}
\label{subsec:rploc}
The main result of this subsection, Corollary~\ref{cor:locality2}, allows us
to reduce the ``global'' evaluation of a first-order formula in a
structure $\cA$ to the ``local'' evaluation of formulas in the
clusters of a neighbourhood cover of $\cA$.
\NewText{To formulate Corollary~\ref{cor:locality2}, we need to 
introduce
the  sets $\FOplus[\myp,q]$ of \cite{GS2026} and
state a result of \cite{GS2026}, namely, \cite[Theorem~7.1]{GS2026}.
For this, we need some more notation.
}

\begin{NewTextBlock}
For $\myp,q\in\NN$ we define $f_q(\myp)\deff (4q)^{q+\myp}$ and
we let $\myrho(\myp,q)\deff f_q(\myp)$.\footnote{\NewText{Let us note that instead of working with the particular function
$f_q(\ell)$, we could work with any function $\myrho$ that satisfies
the constraints (1)--(3) formulated in \cite[Section~4]{GS2026}.
}}
Using this function $\myrho$, we
define sets $\FOplus[\myp,q]\subseteq\FOplus$ for all $\myp,q\in\NN$ with
$\myp\leq q$ as follows.
\begin{itemize}
\item $\FOplus[0,q]$ is the set of all formulas $\phi$ such that
  $|\free(\phi)|\le q$
  and $\phi$ is a Boolean combination of atomic $\FO$ formulas and
  distance atoms $\dist(x,y)\le d$ with $d\le \myrho(0,q)$.
\item For $\myp>0$, $\FOplus[\myp,q]$ is the set of all formulas $\phi$ with
  $|\free(\phi)|\le q{-}\myp$
   such that
   $\phi$ is a Boolean combination of formulas in $\FOplus[\myp{-}1,q]$, distance atoms $\dist(x,y)\le d$ with
  $d\le \myrho(\myp,q)$, formulas $\exists y\, \psi$
  where $\psi\in\FOplus[\myp{-}1,q]$,
  and formulas $\exists y\, \big(\dist(x,y)\le d\, \wedge \psi\big)$,
  where $\psi\in\FOplus[\myp{-}1,q]$
  and $d\le \myrho(\myp,q)-\myrho(\myp{-}1,q)$.
\end{itemize}
Note that every formula $\phi\in\FOplus[\myp,q]$ has quantifier-rank at most $\myp$.
Clearly, for every $\FOplus$-formula $\phi$ there are $\myp,q\in\NN$
with $\myp\leq q$ such that $\phi\in\FOplus[\myp,q]$.

\begin{remark}\label{remark:FOpluspq}
It is easy to see (cf.\ \cite[Remark~4.3]{GS2026}) that
$\FOplus[\myp,q]\subseteq \FOplus[\myp,q']$ for all $\myp,q,q'$ with $\myp\leq
q\leq q'$. On the other hand, if $\myp<\myp'\leq q$, then
$\FOplus[\myp,q]\not\subseteq\FOplus[\myp',q]$ because $\FOplus[\myp,q]$
contains formulas with exactly $q{-}\myp$ free variables, and these
formulas do not belong to $\FOplus[\myp',q]$. But every formula
$\phi\in\FOplus[\myp,q]$ that satisfies $|\free(\phi)|\leq q-\myp'$ also
belongs to $\FOplus[\myp',q]$. 
\end{remark}

For $\myp,q,r,s\geq 0$ with $\myp\leq q$, a \emph{basic local
  $(\myp,q,r,s)$-sentence} is an $\FOplus$ sentence of the form 
\[
  \exists x_1\cdots\exists x_{s'}\;\Bigl(\bigwedge_{1\le i<j\le s'}\dist(x_i,x_j)>2r'\;\wedge\;\bigwedge_{1\le i\le s'}\lambda(x_i)\Bigr),
\]  
where $s'\leq s$, $r'\leq r$, and $\lambda(x)$ is an $r'$-local formula in $\FOplus[\myp,q]$.

We also need the following notation.
Let $\A$ be a $\sigma$-structure, let $k\geq 1$, let
$\ov{a}=(a_1,\ldots,a_k)\in A^k$
and let $r\geq 0$.
By $G^{\A}_{\ov{a},r}$ we denote the graph with vertex
set $[k]$ where there is an edge between nodes $i$ and $j$ iff
$i\neq j$ and $\dist^{\A}(a_i,a_j)\leq r$.
We will often omit the superscript $\A$ and simply write $G_{\ov{a},r}$.

We say that $\ov{a}$ is \emph{$r$-connected} if
the graph $G_{\bar a,r}$ is
connected.  An \emph{$r$-component} of $\bar a$ in $\cA$ is the vertex
set of a connected component of the graph $G_{\bar a,r}$. 

For an arbitrary set $J\subseteq[k]$, by
$\bar a_J$ we denote the projection of $\bar a$ to the positions in
$J$.

We are now ready to state the ``rank-preserving locality theorem for $\FOplus[p,q]$''
of \cite{GS2026} (namely, \cite[Theorem~7.1]{GS2026}).

\begin{theorem}[\cite{GS2026}]\label{thm:locality1}
There is an algorithm which, upon input of $\myp,q\in\NN$ with $\myp< q$ and
a formula $\phi(x_1,\ldots,x_k)\in\FOplus[\myp,q]$ with $k= |\free(\phi)|\geq 1$,
lets $r\deff\myrho(\myp,q)$ and computes
for each $G\in \cG_k$ a number $m_G\in\NN$ and for each $i\in[m_G]$
\begin{itemize}
\item a Boolean combination $\xi^i_G$ of basic local
  $(\myp{-}1,q,\rHalbe,k{+}\myp)$-sentences, and
\item for each connected component $I$ of $G$ an $r$-local formula
  $\psi^i_{G,I}(\bar x_I)\in\FOplus[\myp,q]$
\end{itemize}
such that the following holds.
\begin{enumerate}[(1)]
\item\label{item:localitythm:one}
  For all $\sigma$-structures $\cA$ and all $\bar a\in A^k$ we have
  $\cA\models\phi[\bar a]$ if and only if for $G\deff G_{\bar
    a,r}^{\cA}$ there is an $i\in[m_G]$ such that
  $\cA\models\xi_G^i$ and $\cA\models
  \psi_{G,I}^i[\bar{a}_I]$ for every connected component $I$
  of $G$.
\item\label{item:localitythm:two} 
  For all $\sigma$-structures $\cA$ and all $\bar a\in A^k$, there is
  at most one $i\in [m_G]$ for $G\deff G^{\cA}_{\bar a,r}$ such that  
  $\cA\models\xi_G^i$ and $\cA\models
  \psi_{G,I}^i[\bar{a}_I]$ for every connected component $I$
  of $G$.
\end{enumerate}  
\end{theorem}

Note that if $I\subseteq[k]$ and $\psi(\bar x_I)$ 
is an $r$-local formula, then for every $\sigma$-structure $\cA$,
every neighbourhood cover $\cX$ of $\cA$, every tuple $\bar a\in A^k$,
and all $X,X'\in\cX$ that both $r$-cover $\bar a_I$, we have
\[
  \cA[X]\models\psi[\bar a_I]
  \ \ \iff \ \ 
  \cA\models\psi[\bar a_I]
  \ \ \iff \ \ 
  \cA[X']\models\psi[\bar a_I].
\]
Furthermore, if $I$ is a connected component of $G^{\cA}_{\bar a,r}$
and $\cX$ is a $kr$-neighbourhood cover, then for any entry $a$ of the tuple $\bar a_I$, the cluster $X\deff
\cX(a)$ contains the $kr$-neighbourhood of $a$ in $\A$, and hence $X$
$r$-covers $\bar a_I$.
Using this, as an immediate consequence of Theorem~\ref{thm:locality1}
we obtain the following Corollary~\ref{cor:locality2}. In the
remainder of this paper, Corollary~\ref{cor:locality2} will play the
same role that \cite[Theorem~7.1]{GSarxiv2017,GS2018} played in \cite{GSarxiv2017,GS2018}.

\begin{corollary}[Rank-Preserving Normal Form]
\label{cor:locality2}
  Let $q,k,\ell\in\NN$ such that $1\leq k\le q$ and $\ell\leq q{-}k$, and let $r:=f_q(\ell)$.
  Let
  $\phi(\bar x)$, where $\bar x=(x_1,\ldots,x_k)$, be an
  $\FO^+[\sigma]$-formula in $\FOplus[\ell,q]$.

  Then for each graph $G\in\cG_k$ there are an $m_G\in\NN$ and
  for each $i\in[m_G]$
  \begin{itemize}
  \item a Boolean combination $\xi^i_G$ of basic local $(\ell{-}1,q,\rHalbe,q)$-sentences of signature $\sigma$ and
  \item for each connected component $I$ of $G$ an $r$-local
  $\FO^+[\sigma]$-formula $\psi^i_{G,I}(\bar x_I)\in\FOplus[\ell,q]$
  \end{itemize}
  such that the following holds.
  \begin{enumerate}[(1)]
  \item For all $\sigma$-structures $\cA$, all $kr$-neighbourhood
    covers $\cX$ of $\cA$, and all $\bar a\in A^k$ we have
    $
    \cA\models\phi[\bar a]
    $ 
    if and only if  for
    $G:=G_{\bar a,r}$ there is an $i\in[m_G]$ such
    that $\cA\,\models\,\xi^i_G$ and  for every
     connected component $I$ of $G$ there is an $X\in\cX$ that $r$-covers
    $\bar a_I$ and 
    $
    \cA[X]\ \models\ \psi^i_{G,I}[\bar a_I].
    $
  \item For all $\sigma$-structures $\cA$, all $kr$-neighbourhood
    covers $\cX$ of $\cA$, and all $\bar a\in A^k$ 
    there is at most one $i\in[m_G]$ for $G:=G_{\bar a,r}$ such that 
    the conditions of (1) hold.
  \item For all $\sigma$-structures $\cA$, all $kr$-neighbourhood
    covers $\cX$ of $\cA$, all $\bar a\in A^k$,
    all connected components $I$ of $G:=G_{\bar a,r}$, all
    $X,X'\in\cX$ that
    both $r$-cover $\bar a_I$, and all $i\in[m_G]$,
    \ \(
    \cA[X]\models\psi^i_{G,I}[\bar a_I]
     \iff 
    \cA[X']\models\psi^i_{G,I}[\bar a_I].
    \)
  \item The $\xi^i_G$ and $\psi^i_{G,I}$ can be computed
  from $q,k,\ell,G,\phi$.
  \end{enumerate}
\end{corollary}

\end{NewTextBlock}

\subsection{Rank-preserving term localisation}
Let $\delta_{G,r}(\ov{y})$ be the
$\FO^+[\sigma]$-formula obtained from the $\FO[\sigma]$-formula
$\delta^{\sigma}_{G,r}(\ov{y})$ of Section~\ref{sec:normalform} by
replacing every subformula of the form $\dist^{\sigma}(y_i,y_j)\,{\leq}\,
r$ (resp., ${>}\,r$) by the ``distance atom'' $\dist(y_i,y_j)\,{\leq}\, r$
(resp., its negation).

Next, we define a variant of the cl-terms of Section~\ref{sec:normalform}
that is based on neighbourhood covers. These
``cover-cl-terms'' are no counting terms of the logic
$\FOunC(\Ps)$; they are abstract objects that come with their own
semantics.

\begin{definition}[Cover-cl-Term]\label{def:ccl-term3}
Let $r,m\geq 0$, $k\geq 1$.
 A \emph{basic cover-cl-term} with parameters $(r,k,m)$ and of
 signature $\sigma$ is 
 an object $g$ of the form
 \[
   \Count{(y_1,\ldots,y_k)}{
     \big(\,
      \delta_{G,r}(y_1,\ldots,y_k)
      \;\und\;
      \psi(y_1,\ldots,y_k)
    \,\big) }
 \]
 or an object $u(y_1)$ of the form
 \[
   \Count{(y_2,\ldots,y_k)}{\big(\,
     \delta_{G,r}(y_1,\ldots,y_k)
     \;\und\;
     \psi(y_1,\ldots,y_k)
   \,\big)}
 \]
 where 
 $\ov{y}=(y_1,\ldots,y_k)$ is a tuple of $k$ pairwise distinct
 variables,
 $G$ is a \emph{connected} graph in $\Graphclass_k$, and
 $\psi(y_1,\ldots,y_k)$ is an $\FO^+[\sigma]$-formula 
 such that the following is true for all $\sigma$-structures $\A$,
 all $\ov{a}\in A^k$ with $G^{\A}_{\ov{a},r}=G$, 
 all 
 $m$-neighbourhood covers $\cX$ of $\A$, and all
 clusters $X$ and $X'$ of $\cX$ that $r$-cover $\ov{a}$:
 \[
   \A[X] \ \models \ \psi[\ov{a}]
   \quad\iff \quad
   \A[X'] \ \models \ \psi[\ov{a}]\,.
 \]
 We say that $g$ and $u(y_1)$ 
\NewText{\emph{belong to $\FOplus[\ell,q]$}
  iff $\psi\in\FOplus[\ell,q]$.
}

 \emph{Semantics:}
 For a $\sigma$-structure $\A$ and an $m$-neighbour\-hood cover $\cX$ of
 $\A$ we let
 \ \(
   g^{\A,\cX}
 \) \ 
 be the number of tuples $\ov{a}\in A^k$ such that 
 $G^{\A}_{\ov{a},r}=G$
 (i.e., $\A\models\delta_{G,r}[\ov{a}]$) and
 $\A[X]\models \psi[\ov{a}]$ for some (and hence, all)
 clusters $X$ of $\cX$ that $r$-cover $\ov{a}$.
 Similarly, for $a_1\in A$ we let
 \ \(
   u^{\A,\cX}[a_1]
 \) \ 
 be the number of tuples $(a_2,\ldots,a_k)\in A^{k-1}$ such that for
 $\ov{a}\deff (a_1,a_2,\ldots,a_k)$ we have 
 $G^{\A}_{\ov{a},r}=G$ and 
 $\A[X]\models \psi[\ov{a}]$ for some (hence, all)
 clusters $X$ of $\cX$ that $r$-cover $\ov{a}$.
\smallskip

 A \emph{cover-cl-term with parameters $(r,k,m)$} is built from
 integers and 
 basic cover-cl-terms with parameters $(r',k',\allowbreak m')$ with $r'\leq r$,
 $k'\leq k$, $m'\leq m$ by 
 using rule 
 (\ref{item:plustimesterm}) of Definition~\ref{def:FOC}.
\NewText{We say that the cover-cl-term belongs to $\FOplus[\ell,q]$ iff all 
 its basic cover-cl-terms belong to $\FOplus[\ell,q]$.
}\end{definition}

We generalise the notion to
graphs $G\in\Graphclass_k$ that are not connected.

\begin{definition}[Cover-Term]
Let $r,m\geq 0$, $k\geq 1$.
A \emph{cover-term} with parameters $(r,k,m)$ and of signature
$\sigma$ is of the form
\begin{eqnarray*}
     g & \deff & \Count{(y_1,\ldots,y_k)}{\big(
      \delta_{G,r}(\ov{y})\und \Und_{I\in C} \psi_I(\ov{y}_I)}
     \big) \text{ \ \ or}
\\
     u(y_1) & \deff & \Count{(y_2,\ldots,y_k)}{\big(\,
      \delta_{G,r}(\ov{y})\und \Und_{I\in C} \psi_I(\ov{y}_I)
     \,\big)}
\end{eqnarray*}
where $k\geq 1$, $\ov{y}=(y_1,\ldots,y_k)$ is a tuple of $k$ pairwise
distinct variables, $G\in\Graphclass_k$, $C$ is the set consisting of
all connected components $I$ of $G$, and for every $I\in C$,
$\psi_I(\ov{y}_I)$ is an $\FO^+[\sigma]$-formula 
such that for all $\sigma$-structures $\A$, all
$\ov{a}=(a_1,\ldots,a_k)\in A^{k}$ with
$G^{\A}_{\ov{a}_I,r}=G[I]$,  
all $m$-neighbourhood covers $\cX$ of $\A$, and all clusters $X$ and $X'$ of $\cX$ that
$r$-cover $\ov{a}_I$ we have
\begin{equation}\label{eq:normalform:terms-Version3}
  \A[X] \ \models \ \psi_I[\ov{a}_I]
  \quad\iff \quad
  \A[X'] \ \models \ \psi_I[\ov{a}_I]\,.
  \tag{$**$}
\end{equation}
\emph{Semantics:} For a $\sigma$-structure $\A$ and an
$m$-neighbour\-hood cover $\cX$ of $\A$ we let
$g^{\A,\cX}$ be the number of tuples $\ov{a}=(a_1,\ldots,a_k)\in A^k$ such
that $G_{\ov{a},r}=G$ 
and for all $I\in C$, $\A[X]\models\psi_I[\ov{a}_I]$ for some
(and hence, all) clusters $X$ of $\cX$ that $r$-cover $\ov{a}_I$. 
Furthermore, for every $a_1\in A$ we let 
$u^{\A,\cX}[a_1]$ be the number of tuples $(a_2,\ldots,a_k)\in A^{k-1}$ such
that for $\ov{a}\deff (a_1,a_2,\ldots,a_k)$ we have
$G_{\ov{a},r}=G$ 
and for all $I\in C$, $\A[X]\models\psi_I[\ov{a}_I]$ for some
(and hence, all) clusters $X$ of $\cX$ that $r$-cover $\ov{a}_I$. 
\end{definition}

\begin{lemma}\label{lem:normalform:terms-Version3}
Let $\sigma$ be a relational signature, let $r\geq 0$,
$k\geq 1$, $m\geq 0$, and consider cover-terms
\begin{eqnarray*}
     g & \deff & \Count{(y_1,\ldots,y_k)}{\big(\,
      \delta_{G,r}(\ov{y})\,\und\, \Und_{I\in C} \psi_I(\ov{y}_I)}
     \,\big)
\\
     u(y_1) & \deff & \Count{(y_2,\ldots,y_k)}{\big(\,
      \delta_{G,r}(\ov{y})\,\und\, \Und_{I\in C} \psi_I(\ov{y}_I)
     \,\big)}
\end{eqnarray*}
with parameters $(r,k,m)$.

There exists a ground cover-cl-term $\hat{g}$ and
a unary cover-cl-term $\hat{u}(y_1)$, both 
with parameters $(r,k,m)$, 
such that $ \hat{g}^{\A,\cX} = g^{\A,\cX}$ and
$ \hat{u}^{\A,\cX}[a_1] = u^{\A,\cX}[a_1]$,
for every $\sigma$-structure $\A$, every $m$-neighbourhood cover
  $\cX$ of $\A$, and every $a_1\in A$.

\NewText{For $\ell,q\in\NN$ with $\ell\leq q{-}k$, if $\psi_I(\ov{y}_I)\in\FOplus[\ell,q]$ for each $I\in C$, then 
also $\hat{g}$ and $\hat{u}(y_1)$ belong to $\FOplus[\ell,q]$.
}

Furthermore, there is an algorithm which upon input of $(r,k,m)$, $G$,
and $(\psi_I)_{I\in C}$ constructs $\hat{g}$ and $\hat{u}(y_1)$.
\end{lemma}
\begin{proof}
We proceed by induction on the number $c\deff |C|$ of connected
components of $G$.
Precisely, we show that the following statement $(*)_c$ is true for
every $c\in\NNpos$.
\begin{enumerate}[$(*)_c$:]
\item[$(*)_c$:]
  Let $k\geq c$. 
  Let $G\in\Graphclass_k$ consist of at most $c$ connected components,
  and let $C$ be the set consisting of all connected components $I$ of $G$.
  Let $\ov{y}=(y_1,\ldots,y_k)$ be a tuple of
  $k$ pairwise distinct variables.  Let $r,m\geq 0$.
  For every $I\in C$ let $\psi_I(\ov{y}_I)$ be an
  $\FO^+[\sigma]$-formula such that for all $\sigma$-structures $\A$,
  all $\ov{a}=(a_1,\ldots,a_k)\in A^k$ with
  $G^{\A}_{\ov{a}_I,r}=G[I]$, all $m$-neighbourhood covers $\cX$ of
  $\A$, and all clusters $X$ and $X'$ of $\cX$ that $r$-cover $\ov{a}_I$
  we have
  \[
    \A[X]\ \models\ \psi_I[\ov{a}_I]
    \quad \iff \quad
    \A[X']\ \models\ \psi_I[\ov{a}_I]\,.
  \]
  Then, for the cover-terms $g$ and $u(y_1)$ (as stated in the lemma)
  there are cover-cl-terms $\hat{g}$ and $\hat{u}(y_1)$ with
  parameters $(r,k,m)$ such that
  $g^{\A,\cX}=\hat{g}^{\A,\cX}$ and
  $u^{\A,\cX}[a_1]=\hat{u}^{\A,\cX}[a_1]$ holds for every
  $\sigma$-structure $\A$, every $m$-neighbourhood cover $\cX$ of $\A$
  and all $a_1\in A$.

  \NewText{Furthermore if, for $\ell,q\in\NN$ with $\ell\leq q{-}k$, $\psi_I(\ov{y}_I)\in\FOplus[\ell,q]$ for each $I\in C$, then 
  also $\hat{g}$ and $\hat{u}(y_1)$ belong to $\FOplus[\ell,q]$.
  }\end{enumerate}

\noindent
The induction base for $c=1$ is trivial, since $g$ and $u(y_1)$ are
basic cover-cl-terms.

For the induction step from $c$ to $c{+}1$, consider some $k\geq
c{+}1$ and a graph $G=(V,E)\in\Graphclass_k$ that has $c{+}1$ connected
components. 
Let $V'$ be the connected component of $G$ that contains the node $1$ and
let $V''\deff V\setminus V'$. 

Let $G'\deff\inducedSubStr{G}{V'}$ and $G''\deff\inducedSubStr{G}{V''}$
be the induced subgraphs of $G$ on $V'$ and $V''$, respectively.
Clearly, $G$ is the disjoint union of $G'$ and $G''$, $G'$ is
connected, $G''$ has $c$ connected components, and $C''\deff
C\setminus\set{V'}$ is the set of all connected components of $G''$.

To keep notation simple, we assume (without loss of generality)
that $V'=\set{1,\ldots,k'}$ and $V''=\set{k'{+}1,\ldots,k}$ for
some $k'$ with $1\leq k'<k$.
For any tuple $\ov{z}=(z_1,\ldots,z_k)$ we 
let $\ov{z}{}'\deff (z_1,\ldots,z_{k'})$ and $\ov{z}{}''\deff
(z_{k'+1},\ldots,z_{k})$.

Now consider numbers $r,m\geq 0$ and formulas $\psi_I(\ov{y}_I)$, for
each $I\in C$, as in $(*)_{c+1}$'s assumption.

For every $\sigma$-structure $\A$ and every $m$-neighbourhood cover
$\cX$ of $\A$ we let
\[
  S^{\A,\cX}
\]
be the set of all tuples $\ov{a}=(a_1,\ldots,a_k)\in A^k$ such that
$G^{\A}_{\ov{a},r}=G$ and where for each $I\in C$ we have
$\A[X]\models\psi_I[\ov{a}_I]$ for some (and hence, every) cluster $X$
of $\cX$ that $r$-covers $\ov{a}_I$.
Clearly, $g^{\A,\cX}=|S^{\A,\cX}|$ for the ground cover-term
\[
  g\quad\deff\quad 
  \Count{(y_1,\ldots,y_k)}{\big(\delta_{G,r}(\ov{y})\und \Und_{I\in
      C}\psi_I(\ov{y}_I)\big)}\,.
\]
Similarly, we let
\[
  S_1^{\A,\cX}
\]
be the set of all tuples $\ov{a}'=(a_1,\ldots,a_{k'})\in A^{k'}$ such
that
$G^{\A}_{\ov{a}',r}=G'$ and where $\A[X]\models\psi_{V'}[\ov{a}']$ for
some (and hence, every) cluster $X$ of $\cX$ that $r$-covers $\ov{a}'$.
Clearly, $g_1^{\A,\cX}=|S_1^{\A,\cX}|$ for the \emph{basic cover-cl-term}
\[
  \hat{g}_1\quad\deff\quad 
  \Count{(y_1,\ldots,y_{k'})}{\big(\delta_{G',r}(y_1,\ldots,y_{k'})\und \psi_{V'}(y_1,\ldots,y_{k'})\big)}\,.
\]
\NewText{Note that if, for $\ell,q\in\NN$ with $\ell\leq q{-}k$, $\psi_I(\ov{y}_I)\in\FOplus[\ell,q]$ for each $I\in C$, then 
  also $\hat{g}_1$ belongs to $\FOplus[\ell,q]$.
}

Furthermore, we let
\[
  S_2^{\A,\cX}
\]
be the set of all tuples $\ov{a}''=(a_{k'+1},\ldots,a_k)\in A^{k-k'}$ such that
$G^{\A}_{\ov{a}'',r}=G''$ and where for each $I\in C''$ we have
$\A[X]\models\psi_I[\ov{a}''_I]$ for some (and hence, every) cluster $X$
of $\cX$ that $r$-covers $\ov{a}''_I$.
Clearly, $g_2^{\A,\cX}=|S_2^{\A,\cX}|$ for the ground cover-term
\[
  g_2\quad\deff\quad 
  \Count{(y_{k'+1},\ldots,y_k)}{\big(\delta_{G'',r}(y_{k'+1},\ldots,y_k)\und \Und_{I\in
      C''}\psi_I(\ov{y}_I)\big)}\,.
\]
By the induction hypothesis we know that $(*)_c$ holds. Hence, there
is a \emph{cover-cl-term} $\hat{g}_2$ such that
$\hat{g}_2^{\A,\cX}=g_2^{\A,\cX}=|S_2^{\A,\cX}|$ is true for all
$\sigma$-structures $\A$ and all $m$-neighbourhood covers $\cX$ of
$\A$.
\NewText{Furthermore if, for $\ell,q\in\NN$ with $\ell\leq q{-}k$, $\psi_I(\ov{y}_I)\in\FOplus[\ell,q]$ for each $I\in C$, then 
  also $\hat{g}_2$ belongs to $\FOplus[\ell,q]$.
}

Note that for every $\sigma$-structure $\A$ and every
$m$-neighbourhood cover $\cX$ of $\A$ we have
\[
 S^{\A,\cX}
 \quad = \quad
 \big( S_1^{\A,\cX} \times S_2^{\A,\cX} \big) \
 \setminus\ T^{\A,\cX}\,,
\]
where
\[
  T^{\A,\cX}
\]
is the set of all tuples $\ov{a}=(a_1,\ldots,a_k)\in A^k$ such that 
$\ov{a}'\in S_1^{\A,\cX}$, $\ov{a}''\in S_2^{\A,\cX}$, and there are
$j'\in\set{1,\ldots,k'}$ and $j''\in\set{k'{+}1,\ldots,k}$ such that
$\dist^{\A}(a_{j'},a_{j''})\leq r$.

Let $\GraphclassH$ be the set of all graphs
$H\in\Graphclass_k$ with $H\neq G$, but
$\inducedSubStr{H}{V'}=G'$ and $\inducedSubStr{H}{V''}=G''$.
Clearly, every $H\in\GraphclassH$
has at most $c$ connected components. Furthermore, 
it is straightforward to see that 
for every $\sigma$-structure $\A$ and every $m$-neighbourhood cover
$\cX$ of $\A$, the set $T^{\A,\cX}$ is the disjoint
union of the sets
\[
  T^{\A,\cX}_{H}
  \quad\deff\quad
  \bigsetc{\ 
   \ov{a}\in A^k \ \;}{\ \;
   \ov{a}'\in S_1^{\A,\cX} \ \text{ and } \
   \ov{a}''\in S_2^{\A,\cX} \ \text{ and } \
   G^{\A}_{\ov{a},r}=H
  \ }
\]
for all $H\in\GraphclassH$.

Clearly,
\[
 g^{\A,\cX}
 \quad = \quad
 |S^{\A,\cX}|
  \quad = \quad
  |S_1^{\A,\cX}|\cdot |S_2^{\A,\cX}| \ - \ \sum_{H\in
    \GraphclassH} |T^{\A,\cX}_H|\,;
\]
and this holds for every $\sigma$-structure $\A$ and every
$m$-neighbourhood cover $\cX$ of $\A$. 

To finish the proof of the lemma's statement concerning $g$, it therefore suffices to construct for
each $H\in\GraphclassH$ a cover-cl-term $\hat{g}_H$ such that
$\hat{g}_H^{\A,\cX} = |T^{\A,\cX}_H|$ for every $\sigma$-structure
$\A$ and every $m$-neighbourhood cover $\cX$ of $\A$ --- afterwards,
we are done by choosing
\[
  \hat{g} 
  \quad\deff\quad
  \hat{g}_1 \cdot \hat{g}_2 \ - \ \sum_{H\in\GraphclassH} \hat{g}_H\,.
\]
\NewText{We will ensure that the following is true for $\hat{g}_H$, for each
$H\in\GraphclassH$:
if, for $\ell,q\in\NN$ with $\ell\leq q{-}k$, $\psi_I(\ov{y}_I)\in\FOplus[\ell,q]$ for each $I\in C$, then 
also $\hat{g}_H$ belongs to $\FOplus[\ell,q]$.
Since we already know that $\hat{g}_1$ and $\hat{g}_2$ belong to $\FOplus[\ell,q]$, we
will then obtain that also $\hat{g}$ belongs to $\FOplus[\ell,q]$.
}

Let us consider a fixed $H\in \GraphclassH$. Note that every connected
component of $H$ is a union of one or more connected components of $G$.
Let $I_1,\ldots,I_s$ be the connected components of $H$ (for
$s\leq c$). For each $j\in[s]$ let $C_j$ be the subset of $C$ such
that $I_j=\bigcup_{I\in C_j}I$. W.l.o.g.\ let $V'\in C_1$.

For each $j\in[s]$ let 
\[
  \psi^H_{I_j}(\ov{y}_{I_j})
  \quad\deff\quad
  \Und_{I\in C_j} \psi_I(\ov{y}_I)\,.
\]
It is not difficult to verify that for all $j\in[s]$, 
all $\sigma$-structures $\A$, 
all $\ov{a}=(a_1,\ldots,a_k)\in A^k$ with
$G^{\A}_{\ov{a}_{I_j},r} = H[I_j]$, all $m$-neighbourhood covers $\cX$
of $\A$, and all clusters $X$ and $X'$ of $\cX$ that $r$-cover
$\ov{a}_{I_j}$ we have
\[
  \A[X] \ \models \ \psi^H_{I_j}[\ov{a}_{I_j}]
  \quad\iff \quad
  \A[X'] \ \models \ \psi^H_{I_j}[\ov{a}_{I_j}]\,.
\]
Hence, we can build the cover-term
\[
  g_H 
  \quad \deff \quad
  \Count{(y_1,\ldots,y_k)}{\big(\,
    \delta_{H,r}(\ov{y})
    \,\und\, \Und_{j\in [s]}\psi^H_{I_j}(\ov{y}_{I_j}) 
   \,\big)}
\]
and obtain by the induction hypothesis $(*)_c$ a ground cover-cl-term
$\hat{g}_H$ such that $\hat{g}_H^{\A,\cX}=g_H^{\A,\cX}$ for all
$\sigma$-structures $\A$ and all $m$-neighbourhood covers $\cX$ of $\A$.

\NewText{Note that if for $\ell,q\in\NN$ with $\ell\leq q{-}k$, $\psi_I(\ov{y}_I)\in\FOplus[\ell,q]$ for each $I\in C$, then 
also $\psi^H_{I_j}(\ov{y}_{I_j})\in\FOplus[\ell,q]$ for each $j\in[s]$
(here we use that
$|\free(\psi^H_{I_j}(\ov{y}_{I_j}))|\leq |I_j|\leq k\leq q{-}\ell$, cf.\ Remark~\ref{remark:FOpluspq}). Consequently, also
$\hat{g}_H$ belongs to $\FOplus[\ell,q]$.
}

To finish the proof of the lemma's statement concerning $g$ it remains to show that 
$g_H^{\A,\cX} = |T_H^{\A,\cX}|$.

By definition, $g_H^{\A,\cX}=|U_H^{\A,\cX}|$, for the set
\[
  U_H^{\A,\cX}
\]
of all tuples $\ov{a}=(a_1,\ldots,a_k)\in A^k$ such that
$G^{\A}_{\ov{a},r}=H$ and for all $j\in[s]$,
$\A[X]\models\psi^H_{I_j}[\ov{a}_{I_j}]$ for some (and hence, all)
clusters $X$ of $\cX$ that $r$-cover $\ov{a}_{I_j}$.
It is straightforward to verify that $U_H^{\A,\cX}=T_H^{\A,\cX}$.
This completes the proof of the lemma's statement concerning $g$.

The proof of the lemma's statement concerning $u(y_1)$ follows by an analogous reasoning.
\end{proof}
 
By combining this lemma with
Corollary~\ref{cor:locality2} we obtain the following lemma.

\begin{lemma}[Localisation Lemma]\label{lem:LocalityAndTermsNormalform}
\NewText{Let $q,k,\ell\in\NN$ such that $1\leq k\le q$ and $\ell\leq q{-}k$, and let $r:=f_q(\ell)$.
Let $\phi(\bar x)$, where $\bar x=(x_1,\ldots,x_k)$ for $k$ pairwise distinct
variables $x_1,\ldots,x_k$, be an
$\FO^+[\sigma]$-formula in $\FOplus[\ell,q]$.
}Consider the terms 
\begin{eqnarray*}
  g&\deff&
  \Count{(x_1,\ldots,x_k)}{\phi(x_1,\ldots,x_k)}
\\
 u(x_1)&\deff&
 \Count{(x_2,\ldots,x_k)}{\phi(x_1,\ldots,x_k)}\,.
\end{eqnarray*}
There exists an $s\geq 0$ and
\NewText{basic local $(\ell{-}1,q,\rHalbe,q)$-sentences $\chi_1,\ldots,\chi_s$ of signature $\sigma$
}such that for every
$J\subseteq[s]$ there are a 
ground cover-cl-term $\hat{g}_J$ and a unary cover-cl-term
$\hat{u}_J(x_1)$, both with parameters $(r,k,kr)$,
\NewText{of signature $\sigma$, and belonging to $\FOplus[\ell,q]$,
}such that for every $\sigma$-structure $\A$ and
every $kr$-neighbourhood cover $\cX$ of $\A$
there is exactly one $J\subseteq[s]$ 
with
\[
\A\quad \models \quad 
\chi_J \quad\deff \quad
\Und_{j\in J}\chi_j \ \und \ \Und_{j\in [s]\setminus J} \nicht \chi_j,
\]
and for this $J$ we have
$\hat{g}_J^{\A,\cX} = g^\A$ and
$\hat{u}_J^{\A,\cX}[a] = u^\A[a]$ for every $a\in A$.

Furthermore, there is an algorithm which computes
$\chi_1,\ldots,\chi_s$ and
$\big(\hat{g}_J,\,\hat{u}_J(x_1)\big)_{J\subseteq[s]}$ upon input of
\NewText{$q,k,\ell,\phi(\ov{x})$.
}\end{lemma}

\begin{proof}
We apply Corollary~\ref{cor:locality2} to the formula $\phi(\ov{x})$ and
let $\chi_1,\ldots,\chi_s$ be the list of all
\NewText{basic local $(\ell{-}1,q,\rHalbe,q)$-sentences of signature $\sigma$
}that occur in one of the
$\xi_G^{i}$ for some $G\in\Graphclass_k$ and $i\in[m_G]$.
For every $J\subseteq[s]$ let $S(J)$ be the set of all
$(G,i)$ with
$G\in\Graphclass_k$ and $i\in[m_G]$ for which the propositional
formula obtained from $\xi_G^{i}$ by replacing every occurrence of
$\chi_j$ by $\true$
if $j\in J$ and by $\false$ if $j\not\in J$,
evaluates to $\true$.

For every $G\in\Graphclass_k$ we write $C(G)$ for the set of all
connected components of $G$. 
For every $(G,i)$ with $G\in \Graphclass_k$ and $i\in[m_G]$, consider the
objects 
\begin{align*}
 g_{(G,i)} \ \ \deff \quad &
 \Count{(x_1,\ldots,x_k)}{\big(\,
   \delta_{G,r}(\ov{x})\,\und\, \Und_{I\in C(G)}\psi_{G,I}^{i}(\ov{x}_I)
  \,\big)}
  \shortintertext{and}
  u_{(G,i)}(x_1) \ \ \deff\quad&
 \Count{(x_2,\ldots,x_k)}{\big(\,
   \delta_{G,r}(\ov{x})\,\und\, \Und_{I\in C(G)}\psi_{G,I}^{i}(\ov{x}_I)
  \,\big)}\,.
\end{align*}
From the statement of
Corollary~\ref{cor:locality2} we know that these objects are
cover-terms of signature $\sigma$
\NewText{and with parameters $(r,k,kr)$};
and using Lemma~\ref{lem:normalform:terms-Version3}, we can translate these into
cover-cl-terms $\hat{g}_{(G,i)}$ and $\hat{u}_{(G,i)}(x_1)$ with
parameters $(r,k,kr)$ and
\NewText{belonging to $\FOplus[\ell,q]$.
}

For every $\sigma$-structure $\A$ and every
$kr$-neighbourhood cover $\cX$ of $\A$, there is a unique set
$J\subseteq [s]$ such that
$\A\models\chi_J$.
From the statement of
Corollary~\ref{cor:locality2} we obtain for
\[
  \hat{g}_J
  \quad\deff\quad
  \sum_{(G,i)\in S(J)} \hat{g}_{(G,i)}
\]
and
\[
  \hat{u}_J(x_1)
  \quad\deff\quad
  \sum_{(G,i)\in S(J)}\hat{u}_{(G,i)}(x_1)
\]
that $\hat{g}_J^{\A,\cX}=g^\A$ and
$\hat{u}_J^{\A,\cX}[a]=u^\A[a]$ for every $a\in A$.
\end{proof}

\subsection{The Removal Lemma}\label{subsec:RemovalLemma}

Recall that by $\bar z_I$ we denote the projection of a tuple $\bar z=(z_1,\ldots,z_k)$
to the coordinates in $I\subseteq[k]$. We extend the notation by
letting $\bar z_{\setminus I}:=\bar z_{[k]\setminus I}$.

Let $\sigma$ be a signature and let $r\in\NN$.
\NewText{For every relation symbol $R\in\sigma$
we let $\tilde{R}_\emptyset\deff R$, and for $k\deff \ar(R)$
and for every
set $I\subseteq[k]$
with $I\neq\emptyset$
we introduce
a fresh $(k{-}|I|)$-ary relation symbol $\tilde R_I$.
We 
let $\tilde\sigma$ be
union of $\sigma$ and 
the set of all these new
relation symbols.
}We let $\tilde\sigma_r$ be the extension of
$\tilde\sigma$ by fresh unary relation symbols $S_i$ for all $i\in[r]$.
For every $\sigma$-structure
$\cA$ of order $|A|\ge 2$ and every $d\in A$, we let $\cA\lbag d$ be the
$\tilde\sigma$-structure with universe $A\setminus\{d\}$ and relations
\[
\tilde R_I^{\cA\lbag d}\deff
\big\{\,
  \ov{a}_{\setminus I} \ : \
\ov{a}\in R^{\cA}\text{ and }I=\setc{i\in[k]}{a_i=d}
\,\big\}
\]
for every
\NewText{$R\in\sigma$ and every $I\subseteq[\ar(R)]$.
}Furthermore, we let $\cA\lbag_r d$ be the
$\tilde\sigma_r$-expansion of $\cA\lbag d$ in which each $S_i$ is
interpreted by the set of all $b\in A\setminus\{d\}$ such that
$\dist^{\cA}(d,b)\le i$.
 Note that (for fixed $\sigma$ and $r$), we can compute
$\cA\lbag_r d$ from $\cA$ and $d$ in linear time.

\begin{lemma}[Removal Lemma for Formulas]\label{lem:fo-removal}
\NewText{Let $q,k,\ell\in\NN$ with $\ell\leq q{-}k$, and let
}$r:=f_q(\ell)$.  
Then for every $\FO^+[\sigma]$-formula $\phi(\bar x)$
\NewText{in $\FOplus[\ell,q]$,
}where $\bar
  x=(x_1,\ldots,x_k)$, and for every set $I\subseteq[k]$ there is an
  $\FO^+[\tilde{\sigma}_r]$-formula $\tilde \phi_I(\bar x_{\setminus
    I})$
\NewText{in $\FOplus[\ell,q]$
}such that for all
  $\sigma$-structures $\cA$ of order $|A|\geq 2$, all $d\in A$, and all $\bar
    a=(a_1,\ldots,a_k)\in A^k$ such that 
    $I=\setc{i\in[k]}{a_i=d}$, we
    have
    \begin{equation}\label{eq:RemovalLemma}
    \cA\ \models\ \phi[\bar a]
    \quad\iff\quad
    \cA\lbag_r d\ \models\ \tilde\phi_I[\bar a_{\setminus I}].
    \end{equation}
    Furthermore, there is an algorithm that computes $\tilde\phi_I$
    from
    \NewText{$\phi(\bar x)$ and $I$.}\end{lemma}

\begin{NewTextBlock}
  \begin{proof}
Note that for all distance atoms $\dist(y,z)\leq d'$ that occur in formulas
in $\FOplus[\ell,q]$ it holds that $d'\leq r$.

We proceed by induction on the construction of $\phi$ and show that if
$\phi\in\FOplus[\ell',q]$ for some $\ell'\leq \ell$, then also $\tilde\phi_I\in\FOplus[\ell',q]$.

For the base step, we consider atomic $\FO[\sigma]$-formulas and distance atoms. 
 
If $\phi(x_1,\ldots,x_k)$ is of the form $R(x_{j_1},\ldots,x_{j_s})$
with $s=\ar(R)$ and $j_1,\ldots,j_s\in[k]$, then we let 
  $\tilde\phi_I(\bar x_{\setminus I})\coloneqq \tilde R_J(\bar y)$, where 
$J=\{\nu\in[s] : j_\nu\in I\}$ and
$\bar y$ is the tuple obtained from $(x_{j_1},\ldots,x_{j_s})$ by deleting the entries $x_{j_\nu}$ for all $\nu\in J$.

If $\phi(x_1,\ldots,x_k)$ is of the form $x_{j_1}{=}x_{j_2}$ with $j_1,j_2\in[k]$, then we proceed as follows.
If $j_1,j_2\in I$ then $\tilde{\phi}_I(\bar x_{\setminus I})\coloneqq
\True$\footnote{\NewText{Henceforth, we assume w.l.o.g.\ that the signature contains a
  relation symbol $R_0$ of arity 0, and we let $\true\deff
  R_0()\oder\nicht R_0()$ and $\false\deff\nicht\,\true$. Clearly, $\true$
is valid, $\false$ is unsatisfiable, and both formulas are sentences
that belong to $\FOplus[\ell,q]$ for all $\ell,q$ with $\ell\leq q$.}} (note that in this case $\bar x_{\setminus I}$ does not contain any of the variables $x_{j_1},x_{j_2}$).
If $j_1,j_2\not\in I$ then 
$\tilde{\phi}_I(\bar x_{\setminus I})\coloneqq x_{j_1}{=}x_{j_2}$.
Finally, if $j_1\in I\iff j_2\not\in I$ then $\tilde{\phi}_I(\bar x_{\setminus I})\coloneqq \False$. 

If $\phi(x_1,\ldots,x_k)$ is of the form $\dist(x_{j_1},x_{j_2})\le d'$ with $j_1,j_2\in[k]$, then we proceed as follows.
If $j_1,j_2\in I$ we let $\tilde\phi(\bar x_{\setminus I})\coloneqq\True$.  If $j_1\in I, j_2\not\in I$, we let 
$\tilde\phi_I(\bar x_{\setminus I})\coloneqq S_{d'}(x_{j_2})$;
and analogously if $j_2\in I, j_1\not\in I$, we let
$\tilde\phi_I(\bar x_{\setminus I})\coloneqq S_{d'}(x_{j_1})$. 
Finally, if $j_1,j_2\not\in I$, we let
\[
  \tilde\phi_I(\bar x_{\setminus I}) \ \ \coloneqq\ \ \dist(x_{j_1},x_{j_2})\le {d'}\ \vee\bigvee_{\substack{1\le d_1,d_2\le d'-1,\\d_1+d_2=d'}}\big(S_{d_1}(x_{j_1})\wedge S_{d_2}(x_{j_2})\big).
\]

This completes the construction of $\tilde\phi_I$ for the induction
base. Note that in each case, equation \eqref{eq:RemovalLemma} holds;
and if $\phi(\bar x)\in\FOplus[\ell',q]$ for some $\ell'\leq \ell$, then also $\tilde\phi_I\in\FOplus[\ell',q]$.
\medskip

For the induction step, Boolean combinations are handled in the
obvious way: if $\phi=\nicht\psi$ then $\tilde\phi_I\deff
\nicht\,\tilde\psi_I$, and if $\phi=(\psi\oder\chi)$ then
$\tilde\phi_I\deff (\tilde{\psi}_I\oder\tilde{\chi}_I)$.
Obviously, equation \eqref{eq:RemovalLemma} holds;
and if $\phi(\bar x)\in\FOplus[\ell',q]$ for some $\ell'\leq\ell$, then also $\tilde\phi_I\in\FOplus[\ell',q]$.

Now, all that remains to be done to finish the induction step is to
consider formulas $\phi$ that start with an existential quantifier. We
distinguish between two cases.

\emph{Case 1:} 
$\phi(\bar x)$ is of the form $\exists
x_{k+1}\,\big(\dist(x_i,x_{k+1})\le d'\,\und\,\psi(\bar x')\big)$, where
$i\in[k]$ and $\bar x'=(x_1,\ldots,x_k,x_{k+1})$.
If $i\in I$, we let
   \[ 
    \tilde\phi_I(\bar x_{\setminus I})\ \ \coloneqq \ \ 
    \tilde\psi_{I\cup\{k+1\}}(\bar x'_{\setminus (I\cup\{k+1\})})
    \ \vee \ 
   \exists x_{k+1}\big(S_{d'}(x_{k+1})\wedge\tilde\psi_{I}(\bar x'_{\setminus I})\big).
 \]
 If $i\not\in I$, we let
    \begin{align*}
    \tilde\phi_I(\bar x_{\setminus I})\ \ \coloneqq \ \ &
    \hphantom{\vee \ \ }\big(\, S_d(x_i)\ \wedge \ \tilde\psi_{I\cup\{k+1\}}(\bar x'_{\setminus (I\cup\{k+1\})})\,\big)\\
   & \vee\ \exists x_{k+1}\,\big(\dist(x_i,x_{k+1})\le d\ \wedge \ \tilde\psi_{I}(\bar x'_{\setminus I})\big)\\
    &\vee\bigvee_{\substack{1\le d_1,d_2\le
      d-1\\d_1+d_2=d}}\Big(S_{d_1}(x_i)\ \wedge\ \exists
      x_{k+1}\,\big(S_{d_2}(x_{k+1})\ \wedge\ \tilde\psi_{I}(\bar
      x'_{\setminus I})\big)\Big). 
  \end{align*}
In both cases it can easily be verified that equation
\eqref{eq:RemovalLemma} holds; and if $\phi(\bar x)\in\FOplus[\ell',q]$
for some $\ell'\leq \ell$, then also $\tilde\phi_I\in
\FOplus[\ell',q]$.

\emph{Case 2:}
Case~1 does not apply, and $\phi(\bar x)$ is of the form $\exists
x_{k+1}\,\psi(\bar x')$, with $\bar x'=(x_1,\ldots,x_k,x_{k+1})$.
We let
\[
    \tilde\phi_I(\bar x_{\setminus I})\ \ \coloneqq \ \ 
    \tilde\psi_{I\cup\{k+1\}}(\bar x'_{\setminus(I\cup\{k+1\})})\ \vee \ \exists x_{k+1} \, \tilde\psi_{I}(\bar x'_{\setminus I}).
\]
It is easy to see that equation \eqref{eq:RemovalLemma} holds;  and if
$\phi(\bar x)\in\FOplus[\ell',q]$ for some $\ell'\leq\ell$,
then also $\tilde \phi_I\in\FOplus[\ell',q]$.
This completes the proof of Lemma~\ref{lem:fo-removal}.
\end{proof}
\end{NewTextBlock}

A \emph{basic term} is a term $t(\bar x)$ of the form
$\Count{\bar y}{\phi(\bar x,\bar y)}$ for an $\FO^+$-formula
$\phi(\bar x,\bar y)$.
\NewText{The \emph{width} of $t(\bar x)$ is
$|\bar x|+|\bar y|$.
We say that $t(\bar x)$ \emph{belongs to $\FOplus[\ell,q]$} iff
$\phi(\bar x,\bar y)\in\FOplus[\ell,q]$.
}Usually, we are only interested in ground basic 
terms, where $|\bar x|=0$ and unary basic terms, where
$|\bar x|=1$.

\begin{lemma}[Removal Lemma for Terms]\label{lem:removal}
  Let $\sigma$ be a signature.
\NewText{Let $q,k,\ell\in\NN$ with $\ell\leq q{-}k$, and let
}$r:=f_q(\ell)$.
\begin{enumerate}[(a)]
\item For every ground basic term $g$ of signature $\sigma$, width
  $k$, and
\NewText{belonging to $\FOplus[\ell,q]$,
}there is a list $\hat g_1,\ldots,\hat g_m$ of ground basic
  terms of signature $\tilde\sigma_r$, width at most $k$, and
\NewText{belonging to $\FOplus[\ell,q]$
}such that for all
  $\sigma$-structures $\cA$ of order $|A|\geq 2$ and all $d\in A$,
  \[
  g^{\cA} \quad = \quad \sum_{i=1}^m\hat g_i^{\cA\lbag_r d}
  \]
  Furthermore, there is an algorithm that, given $g$, computes $\hat g_1,\ldots,\hat g_m$.
\item For every unary basic term $u(x)$ of signature $\sigma$, width
  $k$, and
\NewText{belonging to $\FOplus[\ell,q]$,
}there are a list $\hat g_1,\ldots,\hat g_m$ of ground basic
  terms and a list $\hat u_1(x),\ldots,\hat u_n(x)$ of unary basic
  terms, all of signature $\tilde\sigma_r$, width at most $k$, and
\NewText{belonging to $\FOplus[\ell,q]$,
}such that for all
  $\sigma$-structures $\cA$ of order $|A|\geq 2$ and all $a,d\in A$,
  \[
  u^{\cA}[a]\quad =\quad
  \begin{cases}
    \ \ \sum_{i=1}^m\hat g_i^{\cA\lbag_r d}&\text{if }a=d,\\[1ex]
    \ \ \sum_{i=1}^n\hat u_i^{\cA\lbag_r d}[a]&\text{if }a\neq d.
  \end{cases}
  \]
  Furthermore, there is an algorithm that, given $u(x)$, computes
  $\hat g_1,\ldots,\hat g_m, \hat u_1(x),\ldots,\hat u_n(x)$.
\end{enumerate}
\end{lemma}

\begin{proof}
  We only prove assertion (b); the proof of (a) is similar. Let
  \[
  u(x_1)\quad :=\quad \Count{(x_1,\ldots,x_k)}{\phi(x_1,\ldots,x_k)},
  \]
  where $\phi(\bar x)$ is an
  $\FO^+[\sigma]$-formula
\NewText{in $\FOplus[\ell,q]$.
}We apply
Lemma~\ref{lem:fo-removal} to $\phi(\bar x)$ and obtain formulas
$\tilde\phi_I(\bar x_{\setminus I})$ for all $I\subseteq [k]$. Let
$\psi_1(\bar x_1),\ldots,\psi_m(\bar x_m)$ be an enumeration of all
formulas $\tilde\phi_I(\bar x_{\setminus I})$ with $1\in I$, and let 
$\vartheta_1(x_1,\bar x'_1),\ldots,\allowbreak\vartheta_n(x_1,\bar x'_n)$ be an enumeration of all
formulas $\tilde\phi_I(\bar x_{\setminus I})$ with $1\not\in I$. We
let $\hat g_i:=\Count{\bar x_i}{\psi_i(\bar x_i)}$ and $\hat
u_j(x_1):=\Count{\bar x_j'}{\vartheta_j(x_1,\bar x_j')}$.
\end{proof}

\section{Nowhere dense structures} 
\label{sec:nd}

The concept of nowhere dense graph classes tries to capture the
intuitive meaning of ``sparse graphs'' in a fairly general, yet
still useful way. The original definition of nowhere dense classes
(see \cite{nesoss12}), which is relatively complicated, refers to the
edge densities of ``flat minors'' of the graphs in the class. The
definition has turned out to be very robust, and there are are several
seemingly unrelated characterisations of nowhere dense graph
classes. Most useful for us is a characterisation in terms of a the
so-called ``splitter game'' due to \cite{grokresie14}, which we use as
our definition.

Let $G$ be a graph and ${\lrounds},r>0$. 
The~$({\lrounds},r)$-\emph{splitter
  game} on~$G$ is played by two players called \emph{Connector} and
\emph{Splitter} as follows. We let~$G_0:=G$. In round~$i{+}1$ of the
game, Connector chooses an element~$a_{i+1}\in V(G_i)$. Then Splitter
chooses an element $b_{i+1}\in N_r^{G_i}(a_{i+1})$. If $N_r^{G_i}(a_{i+1})\setminus
\{b_{i+1}\}=\emptyset$, then Splitter wins the game. Otherwise, the
game continues with 
\[
  G_{i+1} \quad :=\quad G_i\big[N_r^{G_i}(a_{i+1})\setminus
\{b_{i+1}\}\big]\,.
\]
If Splitter has not won after~${\lrounds}$ rounds, 
Connector wins.

A \emph{strategy} for Splitter is a function~$f$ that associates to
every partial play $(a_1, b_1, \dots, a_i,b_i)$ with associated
sequence~$G_0, \dots, G_i$ of graphs and move~$a_{i+1}\in V(G_i)$ by
Connector a $b_{i+1}\in N_r^{G_i}(a_{i+1})$.  A
strategy~$f$ is a \emph{winning strategy} for Splitter in
the~$({\lrounds},r)$-splitter game on~$G$ if Splitter wins every play in
which she follows the strategy~$f$. If Splitter has a winning strategy,
we say that she \emph{wins} the~$({\lrounds},r)$-splitter game on~$G$.

For a class $\Class$ of graphs and a function $\lambda:\NN\to\NN$, we say that
Splitter wins the \emph{$\lambda$-splitter game on $\Class$} if for
every $r\in\NN$ and every $G\in\Class$ she wins the
$(\lambda(r),r)$-splitter game on $G$.
A class $\Class$ of graphs is \emph{nowhere dense} if there is a
function $\lambda:\NN\to\NN$ such that Splitter wins the $\lambda$-splitter game on $\Class$.
If $\lambda$ is
computable, the class
$\Class$ is \emph{effectively nowhere dense}. A class $\Class$ of 
structures is \emph{(effectively) nowhere dense} if the class of
Gaifman graphs of all structures in $\Class$ is (effectively) nowhere dense.

It follows from \cite{grokresie14} that a class $\Class$ of graphs is
nowhere dense (in the sense just defined) if and only if it is nowhere dense in the
sense of \cite{nesoss12}. 

It is easy to see  that if Splitter wins the
  $(\lrounds,r)$-splitter game on a graph $G$, then she also wins it on
  all subgraphs of $G$. Thus if we close a nowhere dense class of
  graphs under taking subgraphs, the class remains nowhere dense.

Finally, we mention  that for every nowhere dense
class $\Class$ of graphs there is a function $f$ such
that for every $\epsilon >0$ and every
graph $G\in\Class$, if $|V(G)|\ge
f(\epsilon)$
then $\|G\|\le |V(G)|^{1+\epsilon}$ (see
\cite{nesoss12}).

\subsection{Sparse neighbourhood covers}
Let us now turn to \emph{sparse} neighbourhood covers of nowhere dense
graphs. Let $\cX$ be an $r$-neighbourhood cover of a graph $G$ (or of
some structure $\cA$ with Gaifman graph $G$). The \emph{radius} of
$\cX$ is the least $s$ such that all clusters of $\cX$ have radius 
at most $s$, that is, for every $X\in\cX$ there is a $c\in X$ such that
$X\subseteq N_s^{G[X]}(c)$. We call each such $c$ an \emph{$s$-centre}
of $X$. In the following, an \emph{$(r,s)$-neighbourhood cover} of $G$
is an $r$-neighbourhood cover of radius at most $s$.

The \emph{degree} of a vertex~$a\in V(G)$ in a neighbourhood cover $\cX$ 
is the number of clusters $X\in\cX$ such that $a\in
  X$. The
  \emph{maximum degree} $\Delta(\cX)$ is the maximum of the degrees of all
  vertices $a\in V(G)$. Note that
\(
 \sum_{X\in \cX}|X| 
 \ \ \le\ \
 |V(G)|\cdot\Delta(\cX).
\)

\begin{theorem}[\cite{grokresie14}]\label{thm:alg-covers}
  Let~$\Class$ be a nowhere dense class of graphs. Then there is a
  function $f$ and
  an algorithm that, given an $\epsilon>0$, an $r\in\NN$, and a
  graph~$G\in\Class$ with  $n:=|V(G)|\geq f(r,\epsilon)$, computes
  an $(r,2r)$-neighbourhood cover of $G$ of maximum degree at
  most~$n^\epsilon$ in
  time~$f(r,\epsilon)\cdot n^{1+\epsilon}$. Furthermore, if~$\Class$ is
  effectively nowhere dense, then~$f$ is computable.
\end{theorem}

We remark that the construction of \cite{grokresie14} also yields,
together with an $(r,2r)$-neighbourhood cover $\cX$ of $G$, a
function $\cen:\cX\to V(G)$ that associates with each cluster $X\in\cX$, a
$2r$-centre $\cen(X)$ for $X$. Moreover, it is easy to see that for a
given neighbourhood cover $\cX$ of $G$ we can compute in linear time a data
structure that associates with each $X\in\cX$ the list of all $a\in
V(G)$ with $\cX(a)=X$.

\subsection{The main algorithm}

In this section, we complete the proof of
Lemma~\ref{lem:main_FOunC-query-eval}. We fix a numerical predicate
collection $(\Ps,\ar,\sem{.})$ and a signature $\sigma$. 
Let $\Class$
be a  nowhere dense class of structures, and let
$\cG_{\Class}$ be the class of the Gaifman
graphs of all structures in $\Class$. Without loss of
generality we may assume that $\cG_{\Class}$ is closed under taking
subgraphs and that $\Class$ is the class of all
structures whose Gaifman graph is in $\cG_{\Class}$. Let
$\lambda:\NN\to\NN$ such that Splitter wins the $\lambda$-splitter
game on $\cG_{\Class}$.

{\itshape
We need to
design an algorithm with $\Ps$-oracle which receives as input an
$\epsilon>0$, a 
$\sigma$-structure $\A$ from $\Class$
and an $\FOunC(\Ps)[\sigma]$-expression $\xi$ which is
either a sentence $\phi$ or a ground term $t$.
The algorithm decides whether $\A\models\phi$ 
and computes $t^\A$, respectively. The algorithm's running time
is $f(p,\epsilon)n^{1+\epsilon}$, where $p:=\size{\xi}$ is the size of the input expression
and $n:=|A|$ is the order of the input structure.
}

Our algorithm is similar to the model-checking algorithm for
$\FO$-sentences on nowhere dense classes of graphs from
\cite{grokresie14}
\NewText{(also see \cite{GS2026} for a simplified algorithm)}.
The design and analysis of our algorithm relies
on subroutines and results from \cite{grokresie14}. However, we
present a high level outline of the algorithm that should be
accessible without knowledge of \cite{grokresie14}.

The Decomposition Theorem~\ref{thm:normalformFOunC} reduces the
evaluation of $\FOunC(\Ps)[\sigma]$-sentences and ground terms to the
evaluation of first-order sentences and
cl-terms over some signature $\tau\supseteq\sigma$ that extends
$\sigma$ by relation symbols of arity $\leq 1$. Note that every
$\tau$-expansion of a $\sigma$-structure $\cA$ has the same Gaifman
graph as $\cA$ and hence also belongs to $\Class$. The evaluation of
first-order sentences has been taken care of in \cite{grokresie14}.
The evaluation of
cl-terms can further be reduced to basic cl-terms. In fact, it is not
important that we have cl-terms; the important thing is that we have
basic terms with at most one free variable.

So all that remains
is the evaluation of basic terms, either ground terms $g$ or unary
terms $u(x_1)$. To simplify the
notation, we just assume that these terms are in our original
signature $\sigma$. Moreover, we focus on unary terms here; ground
terms can be dealt with similarly.

Hence, the input of our
algorithm is an $\epsilon>0$, a $\sigma$-structure $\cA$ and a
unary basic term 
$u(x_1)$
of width $k$ and
\NewText{belonging to $\FOplus[\ell,q]$ for some $\ell,q\in\NN$ with $\ell\leq q-k$.
}As usual, we let $r=f_q(\ell)$. 
Our algorithm is supposed to compute $u^{\cA}[a]$ for all $a\in A$.

The algorithm proceeds in the following steps.
\begin{enumerate}
\item Let 
 $\delta:=\frac{\epsilon}{2\lambda(2kr)}$. 

If $|A|<f(rk,\delta)$ for the function $f$ of
  Theorem~\ref{thm:alg-covers}, 
  evaluate $t$ by brute force and stop.

  Otherwise, compute a $(kr,2kr)$-neighbourhood co\-ver $\cX$ of
  $\cA$ of maximum degree at most $n^{\delta}$, where $n:=|A|$. In
  addition, compute for
  each $X\in\cX$ a $2kr$-centre $\cen(X)$ and the set of
  all elements $a\in A$ with $\cX(a)= X$.
\item Applying the Localisation Lemma (Lemma~\ref{lem:LocalityAndTermsNormalform}), 
  compute
\NewText{basic local $(\ell{-}1,q,\rHalbe,q)$-sentences
}$\chi_1,\ldots,\chi_s$ and
  cover-cl-terms
  $\big(\hat{g}_J,\,\hat{u}_J(x_1)\big)_{J\subseteq[s]}$ with
  parameters $(r,k,kr)$ and of
\NewText{signature $\sigma$ and belonging to $\FOplus[\ell,q]$
}such that the
  evaluation of $u(x_1)$ in $\cA$ reduces to the evaluation of these
  sentences and terms in $\cA,\cX$.

\item Evaluate the
\NewText{sentences
}$\chi_1,\ldots,\chi_s$ in $\cA$ using the algorithm of \cite{grokresie14}.

 Obviously, there is exactly one set $J\subseteq [s]$ such that
 $\cA\models\chi_J$ for $\chi_J\deff\Und_{j\in J}\chi_j \und\Und_{j\in
   [s]\setminus J}\nicht\chi_j$.
\item\label{item:line5} Compute 
  $\hat{u}_J^{\A,\cX}[a] = u^\A[a]$ for every $a\in A$.
\end{enumerate}
It remains to explain in detail how the last step is carried
out.
Consider a basic cover-cl-term $\hat{u}(x_1)$ that occurs in
$\hat{u}_J(x_1)$ and is of the form
\[
\hat{u}(x_1)
 \ \ := \ \ \Count{(x_2,\ldots,x_{k'})}{\big(\,}
     \delta_{G,{r'}}(x_1,\ldots,x_{k'})
     \;\und
     \psi(x_1,\ldots,x_{k'})
   \,\big)
\]
for a connected graph $G\in\Graphclass_{k'}$, a $k'\le k$, an $r'\leq r$, and an $\FO^+[\sigma]$-formula
$\psi(x_1,\ldots,x_{k'})$
\NewText{in $\FOplus[\ell,q]$.
}

Let $a\in A$ and $X:=\cX(a)$. As $\cX$ is a $kr$-neighbour\-hood cover
and $G$ is connected,
$X$ $r$-covers
every tuple $\bar a=(a_1,\ldots,a_{k'})$ such that $G_{\bar
  a,{r'}}=G$ and $a_1=a$. 
Recall from Definition~\ref{def:ccl-term3} that $\hat{u}^{\A,\cX}[a]$ is
the number of tuples $\ov{a}\in A^{k'}$ such that
$a_1=a$ and 
$G_{\ov{a},r'}=G$ and
$\A[X]\models \psi[\ov{a}]$. 

To be able to compute this number efficiently, we
 introduce a fresh unary relation symbol $Q$ and let $\cB_X$ be the
 $(\sigma\cup\{Q\})$-expansion of $\A[X]$ where $Q$ is
 interpreted by the set of all $a\in A$ such that $\cX(a)=X$. 
Let
 \[
 t(x_1)\ \ :=\ \ \Count{(x_2,\ldots,x_{k'})}{\big(\, }
     \delta_{G,r'}(x_1,\ldots,x_{k'})
     \;\und\;
     \psi(x_1,\ldots,x_{k'}) \;\und\;Q(x_1)
   \,\big)\,.
 \]
\NewText{Note that $\psi(x_1,\ldots,x_{k'})\und Q(x_1)$ is a formula in 
$\FOplus[\ell,q]$, and hence $t(x_1)$ belongs to $\FOplus[\ell,q]$  
(and is of signature $\sigma\cup\set{Q}$).
}What our algorithm needs to do now is evaluate $t(x_1)$ in the structures $\cB_X$,
for all $X\in\cX$. This is done in the following steps, which form the
expanded version of step~\ref{item:line5} of the algorithm.
\begin{enumerate}
\item[\ref{item:line5}.]
  For all $X\in\cX$
  \begin{enumerate}[a.]
  \item Compute $\cB_X$
  \item Let $c:=\cen(X)$, and let $d$ be  Splitter's answer if
    Connector plays $c$ in the 
    first round of the 
$(\lambda(2kr),2kr)$-splitter game on $G_{\cA}$. 
It is explained in
\cite{grokresie14}
\NewText{(cf.\ also \cite[Section~9]{GS2026})}
how $d$ can be computed efficiently.
\NewText{Actually, the efficient computation of Splitter's
  winning strategy in the game is more 
  complicated and requires considering the history of the game; it is
  not sufficient to just look at the first move. For a  detailed
  treatment of this issue we refer the reader to
  \cite{grokresie14,GS2026}.
}\item Compute $\cB':=\cB_X\lbag_rd$.
  \item Apply the Removal Lemma (Lemma~\ref{lem:removal}) to the unary
    basic term $t(x_1)$ and
    recursively evaluate the resulting basic terms in $\cB'$.
  \item 
    For each $a\in Q^{\cB_X}$, use the results of the recursive
    calls to compute $t^{\cB_X}[a]$ according to the Removal Lemma.
  \end{enumerate}
\end{enumerate}
The algorithm terminates with a recursion depth of at most
$\lambda(2kr)$, because in the recursive call we only need to consider
the 
$(\lambda(2kr){-}1,2kr)$-splitter game.

Let us analyse the running time of the algorithm. We express the
running time in terms of the order $n$ of the input structure and the
number $\lrounds$ of rounds of the Splitter game. Initially, we have
$\lrounds=\lambda(2kr)$.
 The dependence on the class $\Class$, the
 signature $\sigma$, and the parameters
\NewText{$k,q,\ell$
}goes into the constants; of course $\lambda(2kr)$ depends on
\NewText{$\Class,k,q,\ell$.
}If $n\le n_0$ for some constant
 $n_0$ (depending on $\Class,\sigma,k,q,\ell,\epsilon$) then the algorithm
terminates in constant time in line $1$. If 
$\lrounds=1$, then
Splitter wins the game in $1$ round, which means that every connected
component of $G_{\cA}$ only consists of a single vertex. Thus either
$|A|=1$ and
the algorithm terminates in 
line~1 in constant time or the algorithm
makes $n$ recursive calls and each of these recursive calls terminates in
constant time. Thus we have the two basic equations
$T(n,\lrounds)=O(1)$ if $n\le n_0$, and 
$T(n,1)=O(n)$ otherwise.

Suppose $n>n_0$ and $\lrounds> 1$. Lines~1--3 can be carried out in time
$O(n^{1+\delta})$ (by Theorem~\ref{thm:alg-covers} for line~1).
To analyse the time spent on line~\ref{item:line5}, let
$X\in\cX$ be of size $n_X:=|X|$. Lines \ref{item:line5}.a--e can be  carried out in time
$O(\|\cB_X\|)=O(\|\cA[X]\|)=O(n_X^{1+\delta})$, because $\cA[X]$ is
from the nowhere dense class $\Class$. The
recursive calls in line \ref{item:line5}.e require time $O(T(n_X,\lrounds{-}1))$. 
Thus the time spent on line~\ref{item:line5} is
\ \(
\sum_{X\in\cX}O\big(T(n_X,\lrounds{-}1)+n_X^{1+\delta}\big),
\) \
and, recalling that 
$\delta=\epsilon/2\lambda(2kr)$ with $\lrounds=\lambda(2kr)$, we obtain a recurrence equation
\[
T(n,\lrounds) \quad = \quad
\sum_{X\in\cX}O\big(T(n_X,\lrounds{-}1)+n_X^{1+\epsilon/2\lrounds}\big)\
+ \ O\big(n^{1+\epsilon/2\lrounds}\big). 
\]
The same recurrence was obtained in \cite{grokresie14}, and it was
shown there that it yields the desired running time
$O(n^{1+\epsilon})$.
This completes our description and
analysis of the algorithm and hence the proof of Lemma~\ref{lem:main_FOunC-query-eval}.

\section{Open questions}\label{sec:conclusion}

To conclude, let us point out some open questions.

(1)~Can our approach be generalised to an extension of $\FO$ which, apart from COUNT, also supports
further aggregate operations of SQL, such as SUM and AVG?

(2)~Can our approach be generalised to support database updates? 
In \cite{KS_LICS17} this was achieved for $\FOCP$ on bounded degree classes. But for other classes, 
e.g., planar graphs or classes of bounded local tree-width (let alone nowhere dense classes), 
this is open even for $\FO$.

\medskip

\NewText{The article \cite{GSarxiv2017} also posed the following question.
}\medskip

(3)~Can our approach be generalised to obtain an algorithm that enumerates the query result
with constant-delay? In \cite{DBLP:conf/icdt/SegoufinV17} such an algorithm was obtained for $\FO$-queries
on classes of locally bounded expansion. Can our machinery of Sections~\ref{sec:normalform} and \ref{sec:Locality}
help to generalise the result to nowhere dense classes?

\medskip

\NewText{This question was answered in \cite{SSV} by providing an algorithm that
enumerates the tuples in the result of an $\FO$-query with
constant-delay after an almost linear time preprocessing phase.
Their solution relied on the flawed ``Theorem~7.1'' of
\cite[Theorem~7.1]{GSarxiv2017,GS2018}; it can be fixed in a
straightforward way by using the new Corollary~\ref{cor:locality2}
instead of \cite[Theorem~7.1]{GSarxiv2017,GS2018}, and
by considering formulas in $\FOplus[\ell,q]$ rather than ``formulas of
$q$-rank at most $\ell$''.
}

\end{document}